\documentclass[12pt]{iopart}
\eqnobysec
\expandafter\let\csname equation*\endcsname\relax
\expandafter\let\csname endequation*\endcsname\relax
\usepackage{amsmath, amsthm, amssymb}
\usepackage{enumerate}
\usepackage{setspace}
\usepackage{natbib}

\newtheorem{theorem}{Theorem}[section]

\newtheorem{proposition}[theorem]{Proposition}

\newenvironment{definition}[1][Definition]{\begin{trivlist}
\item[\hskip \labelsep {\bfseries #1}]}{\end{trivlist}}

\newenvironment{remark}[1][Remark]{\begin{trivlist}
\item[\hskip \labelsep {\bfseries #1}]}{\end{trivlist}}

\newcommand{\derv}[1]{\frac{\partial}{\partial #1}}

\newcommand{\deriv}[2]{\frac{\partial #1}{\partial #2}}

\newcommand{\beqn}{\begin{equation}}
\newcommand{\eeqn}{\end{equation}}
\newcommand{\beqnar}{\begin{eqnarray}}
\newcommand{\eeqnar}{\end{eqnarray}}

\newcommand{\contr}{\,\lrcorner\,}

\begin{document}

\title[Advected Invariants in MHD and Fluids: Noether's Theorems and Casimirs] 
{Local and Nonlocal Advected Invariants and Helicities in Magnetohydrodynamics 
and Gas Dynamics II: Noether's Theorems and Casimirs}
\author{G.M. Webb${}^1$, B. Dasgupta${}^1$, J. F. McKenzie${}^{1,3}$, 
Q. Hu${}^{1,2}$  and
G. P. Zank${}^{1,2}$}
\address{${}^1$ CSPAR, The University of Alabama Huntsville in
 Huntsville AL 35805, USA}

\address{${}^2$ Dept. of Physics, The University of Alabama in Huntsville, Huntsville 
AL 35899, USA}

\address{${}^3$Department of Mathematics and Statistics,\\
Durban University of Technology,\\ Steve Biko Campus, Durban South Africa\\
and School of Mathematical Sciences,\\
University of Kwa-Zulu, Natal, Durban, South Africa}
 
\ead{gmw0002@uah.edu}


\begin{abstract}

Conservation laws in ideal gas dynamics and magnetohydrodynamics (MHD)
associated with fluid relabelling symmetries are derived using 
Noether's first and second theorems. 
Lie dragged invariants are discussed in terms of 
the MHD Casimirs. A nonlocal 
 conservation law for fluid helicity applicable for 
a non-barotropic fluid involving  Clebsch variables is derived 
using Noether's theorem,  
in conjunction with a fluid relabelling symmetry and a gauge 
transformation. A nonlocal cross helicity conservation law involving 
Clebsch potentials, and the MHD energy conservation law are derived 
by the same method. An Euler Poincar\'e variational approach is 
also used to derive conservation laws associated with fluid 
relabelling symmetries using Noether's second theorem. 

\end{abstract}

\submitto{{\it J. Phys. A., Math. and Theor.},28 June, 2013, Revised\ \today} 

\maketitle


\section{Introduction}
In  a recent paper, Webb et al. (2013b) (herein referred to as paper I) 
used Lie dragging techniques (e.g Tur and Janovsky (1993)) 
and Hamiltonian methods using Clebsch variables to investigate 
advected invariants and helicities in ideal fluid mechanics 
and MHD. The main aim of the present paper is to derive some of the 
conservation laws of paper I, by using Noether's theorems 
and gauge transformations, and to relate the invariants obtained by 
the Lie dragging approach to fluid relabelling symmetries and 
the Casimirs of ideal MHD and gas dynamics associated with 
non-canonical Poisson brackets. A conference paper by Webb et al. (2013a)
also studies Lie dragging techniques and advected invariants in 
MHD and fluid dynamics. 

In paper I, we derived the helicity conservation law in fluid dynamics 
and the cross helicity 
conservation law in MHD. In the simplest case of a barotropic equation of 
state for the gas in which the gas pressure $p=p(\rho)$ depends only on 
the gas density one obtains  local conservation laws for 
helicity in fluids and cross helicity in MHD (i.e. 
the conserved densities and fluxes depend only on the 
density $\rho$, the magnetic induction ${\bf B}$, the fluid velocity 
${\bf u}$ and the entropy $S$). For the case of 
cross helicity a local conservation law also holds for a non-barotropic 
equation of state for the gas with $p=p(\rho,S)$ provided the magnetic field 
induction ${\bf B}$ lies in the constant entropy surfaces $S=const.$ 
(i.e. ${\bf B}{\bf\cdot}\nabla S=0$). For the general case of a non-barotropic
equation of state, generalized nonlocal conservation laws 
for helicity and cross helicity were
 obtained by using Clebsch potentials.  
 One of the main aims 
of the present paper is to show how the nonlocal helicity and cross 
helicity conservation laws arise from fluid relabelling symmetries, 
gauge transformations and Noether's theorem. 

The basic MHD model of paper I is described in Section 2. 

Section 3 gives 
a short synopsis of Clebsch variables and Lagrangian and Hamiltonian 
formulations of ideal fluid mechanics and MHD.  Section 3 also gives 
an overview of the MHD Casimirs, i.e. functionals $C$ that have zero 
Poisson bracket with $\{C,K\}=0$ for functionals $K$ dependent on 
the physical variables. There is an overlap in the Casimir 
functionals and the class of functionals that are Lie dragged 
by the flow. 

In Section 4, 
 conservation laws for both barotropic ($p=p(\rho)$) and 
 non-barotropic equations of state
$p=p(\rho,S)$ obtained in paper I are described. 

Section 5 discusses
 Lagrangian MHD and fluid dynamics as developed by Newcomb (1962). 

The Lagrangian approach is used in Section 6, 
to write down the invariance condition for 
the action under fluid relabelling symmetries and gauge
transformations (e.g. Padhye and Morrison 1996a,b).  
We derive the Eulerian version of the invariance condition 
including the effects of gauge transformations, and use 
Noether's theorem to derive the nonlocal helicity and cross 
helicity conservation laws and also the Eulerian energy 
conservation equation, using fluid relabelling symmetries.  

Section 7 uses the Euler-Poincar\'e approach to study the MHD 
equations (e.g. Holm et al. (1998), Cotter and Holm (2012)). 
It shows the important role of the Lagrangian map, which corresponds 
to the Lie group of transformations between the  Lagrangian fluid labels 
and the Eulerian position of the fluid element. The Euler-Poincar\'e 
equation for the MHD system, using Eulerian variations is equivalent to 
the Eulerian MHD momentum equation. The Euler-Poincar\'e approach is used to 
develop Noether's second theorem and the generalized Bianchi identity 
for representative fluid relabelling symmetries. The connection 
of this approach to the more classical approach to Noether's theorem 
of Section 6 is described.  Section 8 concludes with a summary and discussion.  

\section{The Model}
The basic MHD equations used in the model are the same as in paper I. 
The physical quantities $(\rho, {\bf u}^T, p, S, {\bf B}^T)^T$ 
denote the density $\rho$, fluid velcocity ${\bf u}$, gas pressure $p$, entropy $S$, and magnetic field induction ${\bf B}$ respectively. The equations 
 consist of the mass continuity equation, the MHD momentum equation 
written in conservation form using the Maxwell and fluid stress energy 
tensors and the momentum flux $\rho {\bf u}$, the entropy advection 
equation, Faraday's induction equation in the MHD limit, 
 the first law of thermodynamics and Gauss's 
law $\nabla{\bf\cdot}{\bf B}=0$. Faraday's equation, from paper I, can 
be written in the form:
\beqn
\left(\derv{t}+{\cal L}_{\bf u}\right) {\bf B}{\bf\cdot} d{\bf S}
\equiv \left(\deriv{\bf B}{t}-\nabla\times({\bf u}\times {\bf B})
+{\bf u}\nabla{\bf\cdot} {\bf B}\right){\bf\cdot}d{\bf S}=0. 
\label{eq:2.1}
\eeqn
Thus, Faraday's equation corresponds to a conservation law in which the 
magnetic flux ${\bf B}{\bf\cdot}d{\bf S}$ is Lie dragged with the flow, 
where ${\cal L}_{\bf u}={\bf u}{\bf\cdot}\nabla$ is the Lie derivative 
(tangent vector) vector field ${\bf u}$ representing the fluid velocity. 
The first law of thermodynamics for an ideal gas:
\beqn
dQ=TdS=dU+pdV, \label{eq:2.3}
\eeqn
is used where $U$ is the internal energy of the gas per unit mass and 
$V=1/\rho$ is the specific volume of the gas. The internal energy of the 
gas per unit mass is $\varepsilon(\rho,S)=\rho U$. In terms of 
$\varepsilon (\rho,S)$ the first law of thermodynamics gives the equations:
\begin{align}
&\rho T=\varepsilon_S,\quad h=\varepsilon_\rho, \quad 
p=\rho\varepsilon_\rho-\varepsilon, \label{eq:2.4}\\
&-\frac{1}{\rho}\nabla p=T\nabla S-\nabla h, \label{eq:2.5}
\end{align}
where $h$ is the gas enthalpy. 

We also require that Gauss's law $\nabla{\bf\cdot}{\bf B}=0$ is satisfied. 
However, in the Hamiltonian formulation of MHD, 
setting $\nabla{\bf\cdot}{\bf B}=0$ can give rise to problems in ensuring 
that the Jacobi identity is satisfied for all functionals of the 
physical variables (e.g. Morrison and Greene (1980,1982); 
Holm and Kupershmidt (1983a,b); Chandre et al. (2012)). 

\section{Hamiltonian Approach}

In this section we discuss the Hamiltonian approach to MHD and gas dynamics. 
In Section 3.1 we give a brief description of a constrained variational 
principle for MHD using Lagrange multipliers to enforce the constraints 
of mass conservation; the entropy advection equation; Faraday's 
equation and the so-called Lin constraint describing in part, the vorticity
of the flow (i.e. Kelvin's theorem). This leads to Hamilton's canonical 
equations in terms of Clebsch potentials (Zakharov and Kuznetsov (1997)).
In Section 3.2 we transform the canonical 
Poisson bracket obtained from the Clebsch variable approach to a 
non-canonical Poisson bracket written in terms of Eulerian 
physical variables (see e.g. Morrison and Greene (1980,1982), Morrison (1982),  
and Holm and Kupershmidt (1983a,b)). 
In Section 3.3 we obtain the MHD Casimir equations
using the non-canonical variables 
$\boldsymbol{\psi}=({\bf M}, {\bf A}, \rho,\sigma)$ 
where ${\bf M}=\rho {\bf u}$ is the MHD momentum flux, $\sigma=\rho S$  and
${\bf A}$ is the magnetic vector potential in which the gauge is chosen 
so that the 1-form $\alpha={\bf A}{\bf\cdot} d{\bf x}$ is an invariant 
advected with the flow. 

\subsection{Clebsch variables and Hamilton's Equations}

Consider the MHD action (modified by constraints):
\beqn
J=\int\ d^3x\ dt  L,  \label{eq:Clebsch1}
\eeqn
where
\beqnar
L=&&\left\{\frac{1}{2}\rho u^2-\epsilon(\rho, S)-\frac{B^2}{2\mu_0}\right\}
+\phi\left(\deriv{\rho}{t}+\nabla{\bf\cdot}(\rho {\bf u})\right)\nonumber\\
&&+\beta\left(\deriv{S}{t}+{\bf u}{\bf\cdot}\nabla S\right)
+\lambda\left(\deriv{\mu}{t}+{\bf u\cdot}\nabla\mu\right) \nonumber\\
&&+\boldsymbol{\Gamma}{\bf\cdot}\left(\deriv{\bf B}{t}-\nabla\times({\bf u}\times{\bf B})
+{\bf u}(\nabla{\bf\cdot B})\right). \label{eq:Clebsch2}
\eeqnar
The Lagrangian in curly brackets equals the kinetic minus
the potential energy (internal thermodynamic energy plus magnetic energy).
The Lagrange multipliers $\phi$, $\beta$, $\lambda$, 
and $\boldsymbol{\Gamma}$ ensure that the 
mass, entropy, Lin constraint, Faraday equations are satisfied. We do not 
enforce $\nabla{\bf\cdot}{\bf B}= 0$, since we are interested in the 
effect of $\nabla{\bf\cdot}{\bf B}\neq 0$ (which is useful for numerical 
MHD where $\nabla{\bf\cdot}{\bf B}\neq 0$) (see Section 2, and paper I 
for further discussion of this issue). 

 Stationary point conditions for the action are $\delta J=0$.
 $\delta J/\delta {\bf u}=0$ gives the Clebsch representation
for ${\bf u}$:
\beqn
{\bf u}=\nabla\phi-\frac{\beta}{\rho}\nabla S-\frac{\lambda}{\rho}\nabla\mu+
{\bf u}_M\label{eq:Clebsch3}
\eeqn
where
\beqn
 {\bf u}_M=-\frac{(\nabla\times\boldsymbol{\Gamma})\times{\bf B}}{\rho}
-\boldsymbol{\Gamma}\frac{\nabla{\bf\cdot B}}{\rho}, \label{eq:Clebsch4}
\eeqn
is magnetic contribution to ${\bf u}$.
 Setting $\delta J/\delta\phi$, $\delta J/\delta\beta$,
$\delta J/\delta \lambda$, $\delta J/\delta\boldsymbol{\Gamma}$ consecutively 
equal to zero gives the mass, entropy advection, Lin constraint, 
and Faraday (magnetic flux conservation) constraint equations. 
Setting $\delta J/\delta\rho$, $\delta J/\delta S$, $\delta J/\delta\mu$,
$\delta J/\delta {\bf B}$ equal to zero gives evolution equations
for the Clebsch potentials $\phi$, $\beta$ $\lambda$ and $\boldsymbol{\Gamma}$
(see Webb et al. 2013, paper I for details).

 The Hamiltonian functional for the system is given by:
\beqn
{\cal H}=\int H d^3x\quad\hbox{where}\quad H=\frac{1}{2}\rho u^2+\epsilon(\rho,S)+\frac{B^2}{2\mu_0}. \label{eq:H1}
\eeqn
One can show that the evolution equations for 
$(\rho,\phi,{\bf B},\boldsymbol{\Gamma}, S,\beta,\mu,\lambda)$ 
satisfy Hamilton's equations 
for functionals $F$:
\beqn
\dot{F}=\left\{F,H\right\}\quad \hbox{where}\quad \dot{F}=\deriv{F}{t}, 
\label{eq:H3}
\eeqn
and the canonical Poisson bracket  is defined by the equation:
\begin{align}
\{F,G\}=&\int d^3x\ \biggl(\frac{\delta F}{\delta\rho}
\frac{\delta G}{\delta \phi}-\frac{\delta F}{\delta \phi}
\frac{\delta G}{\delta\rho}
+\frac{\delta F}{\delta{\bf B}}
{\bf\cdot}\frac{\delta G}{\delta \boldsymbol{\Gamma}}
-\frac{\delta F}{\delta \boldsymbol{\Gamma}}
{\bf\cdot}\frac{\delta G}{\delta{\bf B}}\nonumber\\
&+\frac{\delta F}{\delta S}
\frac{\delta G}{\delta\beta}-\frac{\delta F}{\delta\beta}
\frac{\delta G}{\delta S}
+\frac{\delta F}{\delta\mu}
\frac{\delta G}{\delta\lambda}-\frac{\delta F}{\delta\lambda}
\frac{\delta G}{\delta\mu}\biggr).
 \label{eq:H4}
\end{align}
Note that  $\{\rho,\phi\}$, $\{S,\beta\}$,
$\{\mu,\lambda\}$, $\{{\bf B},\boldsymbol{\Gamma}\}$
are canonically conjugate variables (see paper I).
The canonical Poisson bracket (\ref{eq:H4})
satisfies the linearity, skew symmetry and Jacobi identity necessary
for a Hamiltonian system (i.e. the Poisson bracket defines a Lie algebra).
\subsection{Non-Canonical Poisson Brackets}
 Morrison and Greene (1980,1982) introduced non-canonical 
Poisson brackets for MHD. {\bf Morrison and Greene (1980) gave the 
non-canonical Poisson bracket for the case $\nabla{\bf\cdot}{\bf B}=0$. 
Morrison and Greene (1982) gave the form of the Poisson bracket in the 
more general case where 
$\nabla{\bf\cdot}{\bf B}\neq 0$. A detailed discussion of the MHD Poisson 
bracket and the Jacobi identity is given in Morrison (1982). Holm and Kupershmidt (1983) point out that their Poisson bracket has the form expected for a semi-direct product Lie algebra, for which the Jacobi identity is automatically 
satisfied.   
 Chandre et al. (2013) use Dirac's theory of constraints to derive 
properties of the Poisson bracket for the $\nabla{\bf\cdot} {\bf B}=0$ case.} 

Introduce the new variables:
\beqn
{\bf M}=\rho {\bf u}, \quad \sigma=\rho S, \label{eq:H5}
\eeqn 
noting that
\beqn
{\bf M}=\rho {\bf u}=\rho\nabla\phi-\beta\nabla S-\lambda\nabla\mu
+{\bf B}{\bf\cdot}(\nabla\boldsymbol{\Gamma})^T
-{\bf B}{\bf\cdot}\nabla\boldsymbol{\Gamma}
-\boldsymbol{\Gamma}(\nabla{\bf\cdot}{\bf B}). \label{eq:H7}
\eeqn
and transforming the canonical Poisson bracket (\ref{eq:H4}) from 
the old variables $(\rho,\phi,S,\beta,{\bf B},\boldsymbol{\Gamma})$ 
to the 
new variables $(\rho,\sigma,{\bf B},{\bf M})$  
we obtain the Morrison and Greene (1982) non-canonical  
Poisson bracket:
\begin{align}
\left\{F,G\right\}=&-\int\ d^3x \biggl\{ \rho
\left[\frac{\delta F}{\delta{\bf M}}
{\bf\cdot}\nabla\left(\frac{\delta G}{\delta\rho}\right)
-\frac{\delta G}{\delta{\bf M}}
{\bf\cdot}\nabla\left(\frac{\delta F}{\delta\rho}\right)\right]\nonumber\\
&+\sigma\left[\frac{\delta F}{\delta{\bf M}}
{\bf\cdot}\nabla\left(\frac{\delta G}{\delta\sigma}\right)
-\frac{\delta G}{\delta{\bf M}}
{\bf\cdot}\nabla\left(\frac{\delta F}{\delta\sigma}\right)\right]\nonumber\\
&+{\bf M}{\bf\cdot}\left[\left(\frac{\delta F}{\delta{\bf M}}
{\bf\cdot}\nabla\right)\frac{\delta G}{\delta{\bf M}}
-\left(\frac{\delta G}{\delta{\bf M}}
{\bf\cdot}\nabla\right)\frac{\delta F}{\delta{\bf M}}\right]\nonumber\\
&+{\bf B}{\bf\cdot}\left[\frac{\delta F}{\delta{\bf M}}
{\bf\cdot}\nabla\left(\frac{\delta G}{\delta{\bf B}}\right)
-\frac{\delta G}{\delta{\bf M}}
{\bf\cdot}\nabla\left(\frac{\delta F}{\delta{\bf B}}\right)\right]\nonumber\\
&+{\bf B}{\bf\cdot}\left[
\left(\nabla\frac{\delta F}{\delta {\bf M}}\right){\bf\cdot}
\frac{\delta G}{\delta{\bf B}}
-\left(\nabla\frac{\delta G}{\delta{\bf M}}\right){\bf\cdot}
\frac{\delta F}{\delta{\bf B}}\right]
\biggr\}. \label{eq:H9}
\end{align}
The bracket (\ref{eq:H9}) has the Lie-Poisson form and 
satisfies the Jacobi identity for all functionals 
$F$ and $G$ of the physical variables, and in general applies both for 
$\nabla{\bf\cdot}{\bf B}\neq 0$  and $\nabla{\bf\cdot}{\bf B}=0$.

\subsubsection{Advected ${\bf A}$ Formulation}
Consider the MHD variational principle using the magnetic vector 
potential ${\bf A}$ instead of using ${\bf B}$ (e.g. Holm and Kupershmidt 
(1983a,b))   
The condition that the magnetic 
flux ${\bf B}{\bf\cdot}d{\bf S}$ is Lie dragged with the flow (i.e. Faraday's 
equation) as a constraint equation, is satisfied if 
 the magnetic 
vector potential 1-form $\alpha={\bf A}{\bf\cdot}d{\bf x}$ is Lie dragged by 
the flow, where 
${\bf B}=\nabla\times{\bf A}$. 
The condition
that ${\bf A}{\bf\cdot}d{\bf x}$ is Lie dragged with the flow implies: 
\beqn
\deriv{\bf A}{t}-{\bf u}\times(\nabla\times{\bf A})+\nabla({\bf u\cdot A})=0 
\label{eq:can1}
\eeqn
(see paper I). 

 We use the variational 
principle $\delta {\cal A}=0$ where the action ${\cal A}$ is given by:
\begin{align}
{\cal A}=&\int_Vd^3x\int dt\biggl\{\left[\frac{1}{2}\rho |{\bf u}|^2 
-\varepsilon(\rho,S)-\frac{|\nabla\times{\bf A}|^2}{2\mu}\right]\nonumber\\
&+\phi\left(\deriv{\rho}{t}+\nabla{\bf\cdot}(\rho{\bf u})\right) 
+\beta \left(\deriv{S}{t}+{\bf u}{\bf\cdot}\nabla S\right)
+\lambda \left(\deriv{\mu}{t}+{\bf u}{\bf\cdot}\nabla \mu\right)\nonumber\\
&+\boldsymbol{\gamma}{\bf\cdot}
\left[\deriv{\bf A}{t}-{\bf u}\times(\nabla\times{\bf A})
+\nabla({\bf u\cdot A})\right]\biggr\}. \label{eq:can2}
\end{align}

By setting the variational derivative $\delta{\cal A}/\delta {\bf u}=0$ gives 
the Clebsch variable expansion:
\beqn
{\bf u}=\nabla\phi-\frac{\beta}{\rho}\nabla S
-\frac{\lambda}{\rho}\nabla\mu-\frac{\boldsymbol{\gamma}\times
(\nabla\times{\bf A})}{\rho}+\frac{\nabla{\bf\cdot}\boldsymbol{\gamma}}{\rho}
{\bf A}, \label{eq:can3}
\eeqn
for the fluid velocity ${\bf u}$. 

In terms of the non-canonical variables $({\bf M},{\bf A},\rho,\sigma)$ 
where $\sigma=\rho S$ we obtain the non-canonical Poisson bracket:
\begin{align}
\left\{F,G\right\}=&-\int\ d^3x\biggl\{ \left[F_{\bf M}{\bf\cdot}
\nabla (G_{\bf M})-G_{\bf M}{\bf\cdot}\nabla(F_{\bf M})\right]
{\bf\cdot}{\bf M}
\nonumber\\
&+\rho\left[F_{\bf M}{\bf\cdot}\nabla(G_{\rho})
-G_{\bf M}{\bf\cdot}\nabla(F_{\rho})\right]\nonumber\\
&+\sigma\left[F_{\bf M}{\bf\cdot}\nabla(G_{\sigma})
-G_{\bf M}{\bf\cdot}\nabla(F_{\sigma})\right]\nonumber\\
&+{\bf A}{\bf\cdot}\left[F_{\bf M}\nabla{\bf\cdot}(G_{\bf A})
-G_{\bf M}\nabla{\bf\cdot}(F_{\bf A})\right]\nonumber\\
&+\nabla\times{\bf A}{\bf\cdot}
\left[G_{\bf A}\times F_{\bf M}-F_{\bf A}\times G_{\bf M}\right]\biggr\}, 
\label{eq:can9}
\end{align}
where $F_{\bf M}\equiv \delta F/\delta {\bf M}$ and similarly for the 
other variational derivatives in (\ref{eq:can9}). The non-canonical 
bracket (\ref{eq:can9}) was obtained by Holm and Kupershmidt (1983a,b). 
It is a skew symmetric bracket and satisfies the Jacobi identity. 
Holm and Kupershmidt (1983a,b) show that bracket (\ref{eq:can9}) 
 corresponds to 
a semi-direct product Lie algebra.

\subsection{The MHD Casimirs}

The Casimirs are defined as functionals that have zero Poisson bracket with 
any functional $K$ defined on the phase space. 
The condition for a Casimir is:
\beqn
\left\{C,K\right\}=0, \label{eq:cas1}
\eeqn
for arbitrary functionals $K$. The Casimirs reveal 
the underlying symmetries of the phase space, implying dependence among the 
variables used to describe the system. The reduced Hamiltonian dynamics, 
taking into account the Casimir constants of motion (note $C_t=0$) 
 takes place
on the symplectic leaves foliating the phase space (e.g. 
Marsden and Ratiu (1994), Morrison (1998), Holm et al. (1998), 
Hameiri (2003,2004)).  

To obtain the Casimir determining equations, 
we introduce the the vector:
\beqn
\zeta=\left(K_{\bf M},K_{\bf A},K_\rho,K_\sigma)=(\boldsymbol{\xi}, 
\boldsymbol{\chi},\lambda,\nu\right), \label{eq:cas2}
\eeqn
where $K_{\bf M}\equiv \delta K/\delta {\bf M}$, and similarly for the other
variational derivatives in (\ref{eq:cas2}). The MHD Poisson bracket 
$\{C,K\}$ can be written in the form:
\begin{align}
\left\{C,K\right\}=&\int
\frac{\delta C}{\delta\psi^a}{\cal A}^{ab}\frac{\delta K}{\delta \psi^b}\ 
d^3x=\int \frac{\delta C}{\delta\psi^a}{\cal A}^{ab}\zeta_b\ d^3x\nonumber\\
=&-\int\zeta_a{\cal A}^{ab}\frac{\delta C}{\delta\psi^b}\ d^3x, \label{eq:cas3}
\end{align}
where   
  $\boldsymbol{\psi}=({\bf M},{\bf A},\rho,\sigma)$. 
The matrix differential operator in (\ref{eq:cas3}) is skew-symmetric, 
i. e. $\left\{C,K\right\}=-\left\{K,C\right\}$. From (\ref{eq:cas3}) 
it follows that for arbitrary $\zeta_b=\delta K/\delta \psi^b$, 
the Casimirs must 
satisfy the  equations:
\beqn
{\cal A}^{ab}\frac{\delta C}{\delta\psi^b}=0. \label{eq:cas4}
\eeqn

\subsubsection{Casimir equations for advected ${\bf A}$}


Using the notation (\ref{eq:cas2}), the gas dynamic terms in 
the bracket (\ref{eq:can9}) are given by:
\begin{align}
&[F_{\bf M}{\bf\cdot}\nabla(G_{\bf M})-G_{\bf M}{\bf\cdot}(\nabla F_{\bf M})]
{\bf\cdot}{\bf M}=[(C_{\bf M}{\bf\cdot}\nabla)\boldsymbol{\xi}
-\boldsymbol{\xi}{\bf\cdot}\nabla(C_{\bf M})]{\bf\cdot}{\bf M}, \nonumber\\
&\rho\left(F_{\bf M}{\bf\cdot}\nabla G_\rho
-G_{\bf M}{\bf\cdot}\nabla F_\rho\right)
=\rho\left(C_{\bf M}{\bf\cdot}\nabla\lambda
-\boldsymbol{\xi}{\bf\cdot}\nabla C_\rho\right), \nonumber\\
&\sigma\left(F_{\bf M}{\bf\cdot}\nabla G_\sigma
-G_{\bf M}{\bf\cdot}\nabla F_\sigma\right)
=\sigma\left(C_{\bf M}{\bf\cdot}\nabla\nu
-\boldsymbol{\xi}{\bf\cdot}\nabla C_\sigma\right),
 \label{eq:acas2}
\end{align}
where $G\equiv K$ and $F\equiv C$. 
Similarly, the magnetic vector potential terms in the 
Poisson bracket (\ref{eq:can9}) are:
\begin{align}
&({\bf A}{\bf\cdot}F_{\bf M})\nabla{\bf\cdot}G_{\bf A}
-({\bf A}{\bf\cdot}G_{\bf M})\nabla{\bf\cdot}F_{\bf A}=
({\bf A}{\bf\cdot}C_{\bf M})\nabla{\bf\cdot}\boldsymbol{\chi}
-({\bf A}{\bf\cdot}\boldsymbol{\xi}) \nabla{\bf\cdot}C_{\bf A}, 
\nonumber\\ 
&{\bf B}{\bf\cdot}\left[G_{\bf A}\times F_{\bf M}
-F_{\bf A}\times G_{\bf M}\right] 
=\boldsymbol{\chi}{\bf\cdot}(C_{\bf M}\times{\bf B})
-\boldsymbol{\xi}{\bf\cdot}({\bf B}\times C_{\bf A}). \label{eq:acas3}
\end{align}
In (\ref{eq:acas2})-(\ref{eq:acas3}) 
${\bf B}=\nabla\times{\bf A}$ and we make the 
identifications $F=C$ and $G=K$. 

Substituting (\ref{eq:acas2})-(\ref{eq:acas3}) in the 
Poisson bracket (\ref{eq:can9}) and integrating the derivative terms 
by parts, and dropping the 
surface terms gives:
\begin{align}
\{C,K\}&=\int\biggl\{-\boldsymbol{\xi}{\bf\cdot}
\left[(\nabla{\bf\cdot}C_{\bf M}){\bf M}
+(C_{\bf M}{\bf\cdot}\nabla){\bf M}
+{\bf M}{\bf\cdot}\left(\nabla C_{\bf M}\right)^T\right]
\nonumber\\
&\quad -\left[\lambda\nabla{\bf\cdot}(\rho C_{\bf M})
+\rho\boldsymbol{\xi}{\bf\cdot}\nabla C_\rho\right]
-\left[\nu\nabla{\bf\cdot}(\sigma C_{\bf M})
+\sigma\boldsymbol{\xi}{\bf\cdot}\nabla C_{\sigma}\right]\nonumber\\
&\quad -\left[\boldsymbol{\chi}{\bf\cdot}\nabla({\bf A}{\bf\cdot} C_{\bf M})
+(\boldsymbol{\xi}{\bf\cdot}{\bf A})\nabla{\bf\cdot} C_{\bf A}\right]
+\boldsymbol{\chi} C_{\bf M}\times{\bf B}
-\boldsymbol{\xi}{\bf\cdot} ({\bf B}\times C_{\bf A})\biggr\} 
\ d^3x\nonumber\\
&=0. 
\label{eq:acas4}
\end{align}

Setting the coefficients of $\lambda$ and $\nu$ equal to 
zero in (\ref{eq:acas4}) gives the equations:
\beqn
\nabla{\bf\cdot}(\rho C_{\bf M})=0,\quad \nabla{\bf\cdot}(\sigma C_{\bf M})=0. 
\label{eq:acas5}
\eeqn
which are analogous to the steady state mass continuity equation 
and entropy conservation  equation with advection velocity 
\beqn
\hat{V}^{\bf x}=C_{\bf M}. \label{eq:acas6}
\eeqn
Setting the coefficient of $\boldsymbol{\chi}$ equal to 
zero in (\ref{eq:acas4}) gives the equation:
\beqn
-C_{\bf M}\times(\nabla\times{\bf A})+\nabla({\bf A}{\bf\cdot}C_{\bf M})=0, 
\label{eq:acas7}
\eeqn
associated with Lie dragging the magnetic vector potential 1-form 
$\alpha={\bf A}{\bf\cdot}d{\bf x}$ with velocity 
$\hat{V}^{\bf x}=C_{\bf M}$. 
Noting that ${\bf M}=\rho {\bf u}$ and setting the coefficient 
of $\boldsymbol{\xi}$ equal to zero in (\ref{eq:acas4}) we obtain the equation:
\beqn
M^k\nabla C_{M^k} +\rho C_{\bf M}{\bf\cdot}\nabla ({\bf M}/\rho)
+\rho \nabla C_\rho 
+\sigma\nabla C_\sigma+{\bf A} (\nabla{\bf\cdot} C_{\bf A}) 
+{\bf B}\times C_{\bf A}=0. \label{eq:acas8}
\eeqn
By noting that for ${\bf B}=\nabla\times {\bf A}$, that 
\beqn
C_{\bf A}=\nabla\times C_{\bf B},\quad \nabla{\bf\cdot} C_{\bf A}=0, 
\label{eq:acas9}
\eeqn
(\ref{eq:acas8}) reduces to:
\beqn
M^k\nabla C_{M^k} +\rho C_{\bf M}{\bf\cdot}\nabla ({\bf M}/\rho)
+\rho \nabla C_\rho
+\sigma\nabla C_\sigma
+{\bf B}\times (\nabla\times C_{\bf B})=0, \label{eq:acas10}
\eeqn
Note that this latter result depends on Gauss's law 
$\nabla{\bf\cdot}{\bf B}=0$ for which ${\bf B}=\nabla\times{\bf A}$. 

Padhye and Morrison (1996a,b) give the Casimir solutions:
\begin{align}
C[\rho,S,{\bf A}]=&\int_V \rho G\biggl(S,\frac{\bf A\cdot B}{\rho},
\frac{{\bf B\cdot}\nabla S}{\rho}, \frac{{\bf B\cdot}\nabla}{\rho}
\left(\frac{{\bf B\cdot}\nabla S}{\rho}\right)\nonumber\\
 &+\frac{{\bf B\cdot}\nabla}{\rho}
\left(\frac{\bf A\cdot B}{\rho}\right), \ldots\biggr)\ d^3x, 
\label{eq:acas11}
\end{align}
It is clear that this family of Casimirs has $C_{\bf M}=0$ and hence the 
gauge dependent condition (\ref{eq:acas7}) does not affect the solution 
of the Casimir determining equations (\ref{eq:acas5}) and (\ref{eq:acas7}). 

The Casimir  (\ref{eq:acas11}) can be related to Lie dragged scalars, 1-forms, 
2-forms, 3-forms and vector fields (e.g. Webb et al. (2013), paper I). 
Let
\begin{align}
{\bf b}=&\frac{\bf B}{\rho}{\bf\cdot}\nabla, \quad 
\alpha=\tilde{\bf A}{\bf\cdot}d{\bf x},\quad 
\nu=\nabla S{\bf\cdot}d{\bf x}, \nonumber\\
\beta=&d\alpha={\bf B}{\bf\cdot} d{\bf S},\quad I=S,\quad \omega=\rho d^3x, 
\label{eq:acas12}
\end{align}
Here ${\bf b}$ is a Lie dragged vector field; $\alpha$ and $\nu$  are  
 1-forms that are Lie dragged with the fluid; 
$\beta={\bf B}{\bf\cdot}d{\bf S}$ is the Lie dragged 
magnetic flux 2-form; $\omega=\rho d^3x$ is a Lie dragged 3-form 
and $I\equiv S$ is an invariant scalar or 0-form  that is advected with 
the fluid (Moiseev et al. (1982), 
Tur and Yanovsky (1993), Webb et al. (2013)). 
Thus 
\begin{align}
&{\bf b}\contr\alpha
=\left(\frac{\bf B}{\rho}{\bf\cdot}\nabla\right)\contr 
(\tilde{\bf A}{\bf\cdot}d{\bf x}= 
 \frac{\tilde{\bf A}{\bf\cdot}{\bf B}}{\rho}, 
\nonumber\\
&{\bf b}\contr\left(\nabla S{\bf\cdot}d{\bf x}\right)=
\left(\frac{\bf B}{\rho}{\bf\cdot}\nabla\right)
\contr\left(\nabla S{\bf\cdot}d{\bf x}\right)
= \frac{{\bf B}{\bf\cdot}\nabla S}{\rho}, \label{acas13}
\end{align} 
and $S$ are invariant, Lie dragged scalars or 0-forms, where the 
symbol $\contr$ denotes the contraction operator 
in the algebra of exterior differential 
forms. Note that the Casimir (\ref{eq:acas11}) is made up of invariant 
Lie dragged forms, and hence the Casimir (\ref{eq:acas11}) is a Lie dragged 
invariant. 

The Casimir equations (\ref{eq:acas5})-(\ref{eq:acas10}) 
obtained by using the Holm and Kupershmidt (1983a,b)  bracket (\ref{eq:can9}) 
are essentially the same as for the Morrison and Greene bracket (see e.g. 
Padhye and Morrison (1996a,b)).  Our main aim here is to show 
that there is a connection between the advected, Lie dragged 
invariants of the MHD system (e.g. Moiseev et al. (1982), 
Tur and Yanovsky (1993), Webb et al. (2013), paper I), 
and the solutions of the Casimir equations. Padhye and Morrison (1996a,b) 
investigate in more detail how the fluid relabelling symmetries 
are related to the Casimirs.  

\section{Helicity Conservation Laws}
In this section we outline the helicity conservation laws obtained in 
paper I. 

\subsection{Fluid Helicities}

In ideal fluid mechanics the helicity transport equation has the form:
\beqn
\deriv{h_f}{t}+\nabla{\bf\cdot}\left[{\bf u} h_f 
+\left(h-\frac{1}{2} |{\bf u}|^2\right)\boldsymbol{\omega}\right]
=\boldsymbol{\omega} T\nabla S+{\bf u}{\bf\cdot}
(\nabla T\times\nabla S), \label{eq:fhel1}
\eeqn
where $\boldsymbol{\omega}=\nabla\times{\bf u}$ is the fluid helicity and 
$h_f={\bf u\cdot}\boldsymbol{\omega}$ is the fluid helicity density.  
For abarotropic gas with $p=p(\rho)$ (\ref{eq:fhel1}) 
implies the helicity conservation law:
\beqn
\deriv{h_f}{t}+\nabla{\bf\cdot}\left[{\bf u} h_f +
\left(h-\frac{1}{2} |{\bf u}|^2\right)
\boldsymbol{\omega}\right]=0, \label{eq:fhel2}
\eeqn

The generalization of the helicity conservation law (\ref{eq:fhel2}) 
for the case of a non-barotropic  equation of state for the gas (i.e
$p=p(\rho, S)$) is given below (cf Proposition 6.1 paper I).

\begin{proposition}\label{propnl1}
The generalized helicity conservation law in ideal fluid mechanics 
can be written in the form:
\beqn
\derv{t}\left[\boldsymbol{\Omega}{\bf\cdot}({\bf u}+r\nabla S)\right] 
+\nabla{\bf\cdot}\left\{{\bf u}\left[\boldsymbol{\Omega}{\bf\cdot}({\bf u}+r\nabla S)\right]+\boldsymbol{\Omega}\left(h-\frac{1}{2}|{\bf u}|^2\right)\right\}
=0. \label{eq:nl1}
\eeqn
The nonlocal conservation law (\ref{eq:nl1}) depends on the Clebsch variable 
formulation of ideal fluid mechanics in which the fluid velocity ${\bf u}$ 
is given by the equation:
\beqn
{\bf u}=\nabla\phi-r\nabla S-\tilde{\lambda} \nabla\mu, \label{eq:nl2}
\eeqn
where $\phi$, $r$, $S$, $\tilde{\lambda}$, and $\mu$ satisfy the equations:
\begin{align}
&\frac{d\phi}{dt}=\frac{1}{2} |{\bf u}|^2-h, \quad \frac{dr}{dt}=-T, 
\nonumber\\
&\frac{dS}{dt}=\frac{d\tilde{\lambda}}{dt}=\frac{d\mu}{dt}=0, 
\label{eq:nl3}
\end{align}
and $d/dt=\partial/\partial t+{\bf u}{\bf\cdot}\nabla$ is the Lagrangian time 
derivative following the flow. In (\ref{eq:nl1}) the generalized vorticity 
$\boldsymbol{\Omega}$ is defined by the equations:
\begin{align}
&{\bf w}={\bf u}-\nabla\phi+r\nabla S\equiv -\tilde{\lambda}\nabla\mu, 
\label{eq:nl4}\\
&\boldsymbol{\Omega}=\nabla\times{\bf w}=\boldsymbol{\omega}
+\nabla r\times\nabla S, \label{eq:nl5}
\end{align}
where $\boldsymbol{\omega}=\nabla\times{\bf u}$ is the fluid vorticity. 
The one-form $\alpha={\bf w}{\bf\cdot}d{\bf x}$ and the two-form 
$\beta=d\alpha=\boldsymbol{\Omega}{\bf\cdot}d S$ 
are advected invariants (see paper I).  
For the barotropic gas case the helicity conservation law (\ref{eq:nl1}) 
reduces  (\ref{eq:fhel2}).
\end{proposition}

\begin{remark}{\bf 1}
The conservation laws (\ref{eq:nl1}) 
is a nonlocal conservation law that involves the nonlocal 
potentials:
\begin{align}
&r({\bf x},t)=-\int_0^t T_0({\bf x}_0,t')\ dt'+r_0({\bf x}_0), 
\label{eq:nl11}\\
&\phi({\bf x},t)=\int_0^t \left(\frac{1}{2}|{\bf u}|^2
-h\right)({\bf x}_0,t')\ dt' 
+\phi_0({\bf x}_0), \label{eq:nl12}
\end{align}
where ${\bf x}={\bf f}({\bf x}_0,t)$ and ${\bf x}_0={\bf f}^{-1}({\bf x},t)$ 
are the Lagrangian map and the inverse Lagrangian map. The temperature 
$T({\bf x},t)=T_0({\bf x}_0,t)$ and $r_0({\bf x}_0)$ and $\phi_0({\bf x}_0)$
are 'integration constants'.
\end{remark}


\begin{proposition}{\bf Ertel's Theorem}
Ertel's theorem in ideal fluid mechanics 
states that the potential vorticity 
$q=\boldsymbol{\omega}{\bf\cdot}\nabla S/\rho$  is a scalar invariant advected
with the flow, i.e., 
\beqn
\frac{d}{dt}\left(\frac{\boldsymbol{\omega}{\bf\cdot}\nabla S}{\rho}\right)=0, 
\label{eq:ertel1}
\eeqn
where $\boldsymbol{\omega}=\nabla\times{\bf u}$ is the fluid vorticity.
In paper I it was shown that there is a higher order invariant, 
the  Hollman invariant $I_h$ involving $I_e$ (see paper I for details).
\end{proposition}

\subsection{MHD Helicities}

We first discuss the magnetic helicity  conservation law, 
followed by a discussion of cross helicity. A more complete 
discussion is given in paper I. 

\subsubsection{Magnetic helicity}
For ideal MHD, the magnetic helicity density 
$h_m={\bf A}{\bf\cdot}{\bf B}$
 satisfies the conservation law:
\begin{equation}
\deriv{h_m}{t}+\nabla{\bf\cdot}
\left[{\bf u}h_m+{\bf B}(\phi_E-{\bf A}{\bf\cdot}{\bf u})\right]=0, 
\label{eq:h2}
\end{equation}
where
\begin{equation}
{\bf E}=-\nabla\phi_E-\deriv{\bf A}{t}=-{\bf u}\times{\bf B},\quad 
{\bf B}=\nabla\times{\bf A}. \label{eq:h3}
\end{equation}
If  
$\tilde{\bf A}={\bf A}+\nabla\Lambda$ where $\Lambda$ is 
the gauge potential for ${\bf A}$ such that 
\beqn
\Lambda=\int^t (\phi_E-{\bf A}{\bf\cdot}{\bf u})\ dt', 
\label{eq:maghel1}
\eeqn
where the integration in (\ref{eq:maghel1}) is with respect to the 
Lagrangian time variable $t'$, then
the magnetic helicity 
conservation law (\ref{eq:h2}) reduces to the advection equation:
\beqn
\deriv{\tilde{h}}{t}+\nabla{\bf\cdot}({\tilde h} {\bf u})=0,
\label{eq:2.32}
\eeqn
where ${\tilde h}={\tilde{\bf A}}{\bf\cdot}{\bf B}$.
is the magnetic helicity density in this special gauge. 

\subsubsection{Cross helicity}
The cross helicity transport equation from paper I, can be written
in the form:
\beqn
\derv{t}\left({\bf u}{\bf\cdot}{\bf B}\right)+\nabla{\bf\cdot}\left[({\bf u}{\bf\cdot}{\bf B}){\bf u}+\left(h-\frac{1}{2} |{\bf u}|^2\right) {\bf B}\right]
=T{\bf B}{\bf\cdot}\nabla S, \label{eq:gch3a}
\eeqn
where $h_C={\bf u}{\bf\cdot}{\bf B}$ is the cross helicity density. If 
${\bf B}{\bf\cdot}\nabla S=0$ the helicity transport equation reduces 
to the cross helicity conservation law. 
 
\begin{proposition}\label{propnl2}
The generalized cross helicity conservation law in MHD can be written in the 
form:
\beqn
\derv{t}\left[{\bf B}{\bf\cdot}\left({\bf u}+r\nabla S\right)\right]
+\nabla{\bf\cdot}\bigl\{ {\bf u}
\left[{\bf B}{\bf\cdot}
\left({\bf u}+r\nabla S\right)
\right]
+\left(h-\frac{1}{2}|{\bf u}|^2\right) {\bf B}\bigr\}=0, 
\label{eq:gch1}
\eeqn
where
\beqn
{\bf u}=\nabla\phi-r\nabla S-\tilde{\lambda}\nabla\mu 
-\frac{(\nabla\times\boldsymbol{\Gamma})\times{\bf B}}{\rho} 
-\boldsymbol{\Gamma}\frac{\nabla{\bf\cdot}{\bf B}}{\rho}, \label{eq:gch2}
\eeqn
is the Clebsch variable representation for the fluid velocity ${\bf u}$, 
and $r({\bf x},t)$ is the Lagrangian temperature integral (\ref{eq:nl11}) 
moving with the flow. 

In the special cases of either (\romannumeral1)\ ${\bf B}{\bf\cdot}\nabla S=0$ 
or (\romannumeral2)\ the case of a barotropic gas with $p=p(\rho)$, the conservation law (\ref{eq:gch1}) reduces to the usual cross helicity conservation law:
\beqn
\derv{t} ({\bf u}{\bf\cdot}{\bf B}) 
+\nabla{\bf\cdot}\left[{\bf u}({\bf u\cdot}{\bf B}) 
+\left(h-\frac{1}{2}|{\bf u}|^2\right) {\bf B}\right]=0, 
\label{eq:gch3}
\eeqn
In general the cross helicity conservation equation (\ref{eq:gch1}) is a 
nonlocal conservation law, in which the variable $r({\bf x},t)$ is a nonlocal 
potential given by (\ref{eq:nl11}).
\end{proposition}

 Detailed proofs of the above helicity and cross helicity conservation 
laws were provided in paper I. In paper I, the concept of topological charge 
was discusssed in relation to advected invariants of the ideal fluid and MHD 
equations (see also Kamchatnov (1982) and Semenov et al. (2002)). 
 The physical application of magnetic helicity 
in solar,  space  and fusion plasmas is discussed by 
Berger and Field (1984), Finn and Antonsen (1985,1988), 
Berger and Ruzmaikin (2000), Bieber et al. (1987), 
Low (2006), Longcope and Malunushenko (2008), Yahalom and Lynden Bell (2008), 
Yahalom (2013) 
and Webb et al. (2010a,b). Tur and Janovsky (1993) and 
Webb et al. (2013a,b)  discuss the Godbillon Vey invariant
 which applies 
for MHD flows with zero magnetic helicity, i.e. 
$\tilde{\bf A}{\bf\cdot}\nabla\times\tilde{\bf A}=0$, 
where $\alpha=\tilde{\bf A}{\bf\cdot} d{\bf x}$  is 
Lie dragged with the flow 
and ${\bf B}=\nabla\times\tilde{\bf A}$. Kats (2003) obtains the 
MHD version of the Ertel invariant.

\section{The Lagrangian map} 

\subsection{Lagrangian MHD} 
The Lagrangian map: ${\bf x}={\bf X}({\bf x}_0,t)$ is obtained by integrating 
the fluid velocity equation $d{\bf x}/dt={\bf u}({\bf x},t)$, subject to 
the initial condition ${\bf x}={\bf x}_0$ at time $t=0$. This approach 
to MHD was initially developed by Newcomb (1962). 
 In this approach, the 
In Lagrangian MHD, the mass continuity equation and 
entropy advection equation are replaced 
by the equivalent algebraic equations:
\beqn
\rho={\frac{\rho_0({\bf x}_0)}{J}}, \quad S=S({\bf x}_0), \label{eq:2.7n}
\eeqn
where
\beqn
J=\det(x_{ij})\quad\hbox{and}\quad x_{ij}={\deriv{x^i({\bf x}_0,t)}{x_0^j}}. 
\label{eq:2.8n}
\eeqn
Similarly, Faraday's equation (\ref{eq:2.4}) has the formal solution for the
magnetic field induction ${\bf B}$ of the form:
\beqn
B^i={\frac{x_{ij}B_0^j}{J}},\quad \nabla_0{\bf\cdot}{\bf B}_0=0. \label{eq:2.9n}
\eeqn
The solution (\ref{eq:2.9n}) for $B^i$ is equivalent to the frozen in 
field theorem in MHD (e.g. Parker (1979)), and the 
initial condition $\nabla_0{\bf\cdot}{\bf B}_0=0$
is imposed in order to ensure that Gauss's law $\nabla{\bf\cdot}{\bf B}=0$ 
is satisfied. 

The Lagrangian map ${\bf x}={\bf X}({\bf x}_0,t)$ and its inverse 
${\bf x}_0={\bf X}_0({\bf x},t)$ is discussed in detail in Newcomb (1962),
Webb et al. (2005b), Webb and Zank (2007) and others.  One can 
 show that the Lagrange labels ${\bf x}_0$ are advected with the flow. 


The action for the MHD system is:
\beqn
A=\int \int {\cal L}\ d^3x dt\equiv \int\int {\cal L}^0\  d^3x_0 dt, 
\label{eq:2.15n} 
\eeqn
where
\beqn
{\cal L}={\frac{1}{2}}\rho |{\bf u}|^2-\varepsilon(\rho,S)-{\frac{B^2}{2\mu}}
-\rho\Phi,\quad {\cal L}^0={\cal L} J, \label{eq:2.16n}
\eeqn
are the Eulerian (${\cal L}$) and Lagrangian (${\cal L}^0$) Lagrange densities
respectively. Using (\ref{eq:2.7n})-(\ref{eq:2.9n}), and (\ref{eq:2.16n}) 
we obtain:
\beqn
{\cal L}^0={\frac{1}{2}}\rho_0 |{\bf x}_t|^2
-J\varepsilon\left({\frac{\rho_0}{J}},S\right)
 -{\frac{x_{ij}x_{is} B_0^j B_0^s}{2\mu J}}-\rho_0\Phi, \label{eq:2.17n}
\eeqn
for ${\cal L}^0$. Note that in 
${\cal L}^0={\cal L}^0({\bf x}_0,t; {\bf x}, {\bf x}_t, x_{ij})$, 
${\bf x}_0$ and $t$ are the independent variables, and ${\bf x}$ 
and its derivatives with respect to ${\bf x}_0$ and $t$ are dependent
variables.

The Hamiltonian description of MHD using the Lagrangian map 
is described by Newcomb (1962) (see also Padhye and Morrison (1996), 
Webb et al. (2005b), Webb and Zank (2007)).

\section{Symmetries and Noether's theorem in MHD}
In this section we discuss Noether's first theorem in MHD (e.g. 
 Padhye (1998),  Webb et al. (2005b)). We 
consider the Lagrangian action (\ref{eq:2.15n}), namely
\beqn
A=\int\int {\cal L}^0\  d^3x_0 dt,
\label{eq:3.1n}
\eeqn
where the Lagrangian density ${\cal L}^0$ is given by (\ref{eq:2.17n}). 

\subsection{Noether's theorem}

\begin{proposition}{Noether's theorem}\label{prop8.1}.
If the action (\ref{eq:3.1n}) is invariant to $O(\epsilon)$ under the 
infinitesimal Lie transformations:
\beqn
x'^i=x^i+\epsilon V^{x^i}, \quad
x'^j_0=x^j_0+\epsilon V^{x_0^j}, \quad
t'=t+\epsilon V^t, \label{eq:3.2n}
\eeqn
and under the divergence transformation:
\beqn
{\cal L}^{0'}={\cal L}^0+\epsilon D_\alpha \Lambda_0^\alpha+O(\epsilon^2), 
\label{eq:3.3n}
\eeqn
(here $D_0\equiv \partial/\partial t$ and $D_i\equiv \partial/\partial x_0^i$
are the total derivative operators in the jet-space consisting 
of the derivatives of $x^k({\bf x}_0,t)$ and physical quantities 
that depend on ${\bf x}_0$ and $t$) then the MHD system admits the 
Lagrangian conservation law:
\beqn
{\deriv{I^0}{t}}+{\deriv{I^j}{x_0^j}}=0, \label{eq:3.4n}
\eeqn
where
\beqnar
&&I^0=\rho_0 u^k {\hat V}^{x^k}+V^t {\cal L}^0+\Lambda_0^0, 
\label{eq:3.5n}\\
&&I^j={\hat V}^{x^k} \left[\left( p+{\frac{B^2}{2\mu}}\right) \delta^{ks}
-{\frac{B^k B^s}{\mu}}\right] A_{sj} + V^{x_0^j} {\cal L}^0 +\Lambda_0^j,
\label{eq:3.6n}
\eeqnar
In (\ref{eq:3.5n})-(\ref{eq:3.6n})
\beqn
{\hat V}^{x^k({\bf x}_0,t)} = V^{x^k({\bf x}_0,t)}-\left(V^t{\derv{t}}
+V^{x_0^s} {\derv{x_0^s}} \right)x^k({\bf x}_0,t), \label{eq:3.7n}
\eeqn
is the canonical or evolutionary Lie symmetry transformation generator 
corresponding
to the Lie transformation (\ref{eq:3.2n})
(i.e. $x'^k=x^k+\epsilon {\hat V}^{x^k}$, $t'=t$, $x'^j_0=x_0^j$).
\end{proposition}
 Proof of the above form of Noether's theorem for MHD is given in Webb et al. (2005b) and in Webb and Zank (2007).
General proofs of  Noether's first theorem are
given in Bluman and Kumei (1989) and Olver (1993).

\begin{remark}\label{prop8.0}
The action (\ref{eq:3.1n}) is invariant to $O(\epsilon)$
under the divergence transformation of the form (\ref{eq:3.2n})-(\ref{eq:3.3n})
provided:
\beqn
{\tilde X} {\cal L}^0+{\cal L}^0 \left(D_t V^t+D_{x_0^j} V^{x_0^j}\right)
+D_t \Lambda_0^0+D_{x_0^j} \Lambda_0^j=0, \label{eq:3.9n}
\eeqn
where
\beqn
{\tilde X}=V^t {\derv{t}}+V^{x_0^s}\derv{x_0^s}
+V^{x^k}{\derv{x^k}}+V^{x^k_t} {\derv{x^k_t}}
+V^{x_{kj}} {\derv{x_{kj}}}+\cdots, \label{eq:3.10n}
\eeqn
is the extended Lie transformation operator.
Here ${\tilde X}$
gives the transformation rules for the derivatives of $x^k({\bf x}_0,t)$
under Lie transformation (\ref{eq:3.2n}). From Ibragimov (1985):
\beqnar
&&{\tilde X}={\hat X}+V^\alpha D_\alpha, \label{eq:3.11n}\\
&&{\hat X}= {\hat V}^{x^k} {\derv{x^k}}+ D_\alpha \left({\hat V}^{x^k}\right)
{\derv{x^k_\alpha}}
+D_\alpha D_\beta\left({\hat V}^{x^k}\right){\derv{x^k_{\alpha\beta}}}
+\cdots, \label{eq:3.12n}
\eeqnar
where $D_0=\partial/\partial t$ $D_i=\partial/\partial x_0^i$ denote
total partial derivatives with respect to $t$ and $x_0^i$
($1\leq i\leq 3$), $V^0\equiv V^t$ and $V^i\equiv V^{x_0^i}$
respectively. ${\hat X}$ is the extended Lie symmetry operator for
 the canonical Lie transformation $x'^k=x^k+\epsilon {\hat V}^{x^k}$,
$t'=t$ and $x'^j_0=x^j_0$.
\end{remark}

\begin{proposition}\label{prop8.2}
The Lagrangian conservation law (\ref{eq:3.4n}) can be written as an 
Eulerian conservation law of the form (Padhye (1998)):
\beqn
{\deriv{F^0}{t}}+{\deriv{F^j}{x^j}}=0, \label{eq:3.14n}
\eeqn
where
\beqn
{F^0}={\frac{I^0}{J}},\quad F^j={\frac{u^j I^0+x_{jk}I^k}{J}}, 
\quad (j=1,2,3), \label{eq:3.15n}
\eeqn
are the conserved density $F^0$ and flux components $F^j$.
\end{proposition}

\begin{proposition}\label{prop8.3}
The Lagrangian conservation law (\ref{eq:3.4n}) with conserved 
density $I^0$ of (\ref{eq:3.5n}), and flux $I^j$ of (\ref{eq:3.6n}), 
is equivalent to the 
Eulerian conservation law:
\beqn
{\deriv{F^0}{t}}+{\deriv{F^j}{x^j}}=0, \label{eq:3.16n}
\eeqn
where
\beqnar
&&F^0=\rho u^k {\hat V}^{x^k({\bf x}_0,t)} +V^t {\cal L}+\Lambda^0, 
\label{eq:3.17n}\\ 
&&F^j={\hat V}^{x^k({\bf x}_0,t)}\left(T^{jk}-{\cal L}\delta^{jk}\right)
+V^{x^j}{\cal L}+\Lambda^j, \label{eq:3.18n}\\
&&T^{jk}=\rho u^j u^k +\left(p+{\frac{B^2}{2\mu}}\right)\delta^{jk}
-{\frac{B^jB^k}{\mu}}, \label{eq:3.19n}\\
&&\Lambda^0={\frac{\Lambda_0^0}{J}},
\quad \Lambda^j={\frac{u^j \Lambda_0^0+x_{js}\Lambda_0^s}{J}}. \label{eq:3.20n}
\eeqnar
In (\ref{eq:3.16n})-(\ref{eq:3.20n}) $T^{jk}$ is the Eulerian momentum 
flux tensor (the spatial components of the stress energy tensor) 
and ${\hat V}^{x^k({\bf x}_0,t)}$ is the canonical symmetry generator 
(\ref{eq:3.6n}). 
\end{proposition}

\begin{remark}
For a pure fluid relabelling symmetry $V^{\bf x}=V^t=0$, and 
Proposition \ref{prop8.3} gives:
\begin{align}
F^0=&\hat{V}^{\bf x}{\bf\cdot}(\rho {\bf u})+\Lambda^0, \label{eq:3.20na}\\
{\bf F}=&\hat{V}^{\bf x}{\bf\cdot}\left[\rho {\bf u}\otimes{\bf u}
+\left(\varepsilon+p+\rho\Phi +\frac{B^2}{2\mu_0}\right)\sf{I}
-\frac{{\bf B}\otimes{\bf B}}{\mu_0}\right] +\boldsymbol{\Lambda}, 
\label{eq:3.20nb}
\end{align}
for the conserved density $F^0$ and flux ${\bf F}$ where 
$\boldsymbol{\Lambda}=(0,\Lambda^1,\Lambda^2,\Lambda^3)$.
\end{remark}

\begin{remark} 
Padhye and Morrison (1996a,b) and Padhye (1998) used 
Proposition \ref{prop8.2} to convert 
Lagrangian conservation laws to Eulerian conservation laws.  
 Webb et al. (2005b) derived Lagrangian 
and Eulerian conservation laws  
 using Propositions \ref{prop8.1} and \ref{prop8.3}, 
and studied   
 fully nonlinear MHD waves in a non-uniform and time dependent
 background flow. Linear waves 
in a non-uniform background were studied in Webb et al. (2005a), 
 extending similar work by Dewar (1970) for WKB waves. 
\end{remark}

\subsection{Fluid Relabelling Symmetries}

Consider infinitesimal Lie transformations of the form 
(\ref{eq:3.2n})-(\ref{eq:3.3n}), with
\beqn
V^t=0,\quad V^{\bf x}=0,\quad V^{{\bf x}_0}\neq 0,  
\label{eq:3.21n}
\eeqn
which leave the action (\ref{eq:3.1n}) invariant. The extended Lie 
transformation operator ${\tilde X}$ for the 
case (\ref{eq:3.21n}) has generators:
\beqnar
&&{\hat V}^{\bf x}=-V^{{\bf x}_0}{\bf\cdot}\nabla_0 {\bf x},\quad 
V^{{\bf x}_t}=-D_t \left(V^{{\bf x}_0}\right) {\bf\cdot}\nabla_0 {\bf x},
\nonumber\\
&&V^{\nabla_0{\bf x}}=-\nabla_0\left(V^{{\bf x}_0}\right)
{\bf\cdot}\nabla_0 {\bf x}. 
\label{eq:3.22n}
\eeqnar
The condition (\ref{eq:3.9n}) for a divergence symmetry of the action 
reduces to:
\beqnar
&&\nabla_0{\bf\cdot}\left( \rho_0 V^{{\bf x}_0}\right) 
\left( {\frac{1}{2}}|{\bf u}|^2-\Phi ({\bf x})
-{\frac{\varepsilon+p}{\rho}}\right)
 -J{\deriv{\varepsilon(\rho,S)}{S}}
V^{{\bf x}_0}{\bf\cdot}\nabla_0 S\nonumber\\
&&-D_t\left(\rho_0 V^{{\bf x}_0}\right){\bf\cdot}\nabla_0{\bf x}{\bf\cdot}
{\bf u}
-{\frac{1}{\mu J}}(\nabla_0 {\bf x}){\bf\cdot}(\nabla_0{\bf x})^T
{\bf :}\bigl[\bigl(V^{{\bf x}_0}{\bf\cdot}\nabla_0 {\bf B}_0\bigr) {\bf B}_0
\nonumber\\
&&+{\bf B}_0{\bf B}_0 \nabla_0{\bf\cdot} V^{{\bf x}_0} 
-\left( {\bf B}_0{\bf\cdot}\nabla_0 V^{{\bf x}_0}\right) {\bf B}_0\bigr]=
-\partial\Lambda_0^\alpha/\partial x_0^\alpha.
\label{eq:3.23n}
\eeqnar
where $x_0^\alpha=(t,x_0,y_0,z_0)$ is the spatial four-vector in Lagrange label 
space. 
 Simple solutions of (\ref{eq:3.23n}) with $\Lambda_0^\alpha=0$ 
($\alpha=0,1,2,3$) are obtained by setting:  
\beqnar
&&\nabla_0{\bf\cdot}\left( \rho_0 V^{{\bf x}_0}\right)=0,\quad 
V^{{\bf x}_0}{\bf\cdot}\nabla_0 S=0,
\quad D_t\left(\rho_0 V^{{\bf x}_0}\right)=0, \nonumber\\
&&\nabla_0\times\left(V^{{\bf x}_0}\times {\bf B}_0\right)=0,\quad 
\nabla_0{\bf\cdot}{\bf B}_0=0, \quad \Lambda_0^\alpha=0, \label{eq:3.24n}
\eeqnar
where $\alpha=0,1,2,3$. 
Equations (\ref{eq:3.24n}) are Lie determining equations for the fluid 
relabelling symmetries obtained by Padhye (1998) and Webb et al. (2005b).
However, (\ref{eq:3.24n}) do not give the most general 
solutions for the fluid relabelling symmetries. To obtain other 
possible solutions of (\ref{eq:3.23n}) it is useful to convert the 
fluid relabelling divergence symmetry  condition to its  
Eulerian form given below. 

\begin{proposition}\label{prop8.4}
The condition (\ref{eq:3.23n}) for a divergence symmetry of the action 
converted to Eulerian form is:
\begin{align}
&\nabla{\bf\cdot}\left(\rho\hat{V}^{\bf x}\right) 
\left(h+\Phi({\bf x})-\frac{1}{2}|{\bf u}|^2\right) 
+\rho T\hat{V}^{\bf x}{\bf\cdot}\nabla S
+\rho {\bf u}{\bf\cdot}\left(\frac{d\hat{V}^{\bf x}}{dt}
-\hat{V}^{\bf x}{\bf\cdot}\nabla{\bf u}\right)\nonumber\\
&+\frac{\bf B}{\mu_0}{\bf\cdot}\left[-\nabla\times
\left(\hat{V}^{\bf x}\times{\bf B}\right) 
+\hat{V}^{\bf x}\nabla{\bf\cdot}{\bf B}\right]
=-\nabla_\alpha\Lambda^\alpha, 
\label{eq:3.25n}
\end{align}
where 
\beqn
\nabla_\alpha\Lambda^\alpha=\deriv{\Lambda^0}{t}+\deriv{\Lambda^i}{x^i}, 
\label{eq:3.26n}
\eeqn
is the four divergence of the four dimensional vector 
$\boldsymbol{\Lambda}= \left(\Lambda^0,\Lambda^1,\Lambda^2,\Lambda^3\right)$. 
The four vector $\boldsymbol{\Lambda}$ is related the the Lagrange 
label space vector $\Lambda_0^\alpha$ by the transformations:
\beqn
\Lambda^\alpha=\Lambda_0^\beta B_{\beta\alpha}\equiv \Lambda^\beta_0 
\frac{x_{\alpha \beta}}{J}, \label{eq:3.27n}
\eeqn
where $x_{\alpha\beta}=\partial x^\alpha/\partial x_0^\beta$, $J=\det(x_{ij})$ 
and $B_{\alpha\beta}=\hbox{cofac}(\partial x_0^\alpha/\partial x^\beta)$ (the 
transformations (\ref{eq:3.27n}) are the same as those in (\ref{eq:3.20n}); 
note that $\alpha$, $\beta$ have values $0,1,2,3$). 
\end{proposition}

\begin{proof}
The proof follows by using (\ref{eq:3.1n})-(\ref{eq:3.7n}) and 
 the transformations (\ref{eq:3.22n}) 
relating $\hat{V}^{\bf x}$, $\hat{V}^{{\bf x}_t}$ and 
$\hat{V}^{{\bf x}_{ij}}$  to $V^{{\bf x}_0}$.
\end{proof}

The divergence symmetry conditions (\ref{eq:3.25n}) and Noether's theorem 
(Proposition \ref{prop8.3}) 
applied to the fluid relabelling symmetries, and including 
the gauge potentials $\Lambda^\alpha$ ($\alpha=0,1,2,3$) are used below to 
derive the nonlocal fluid helicity conservation law (\ref{eq:nl1}) 
and the nonlocal cross helicity conservation law (\ref{eq:gch1}). 

\begin{proposition}\label{prop8.5}
The fluid helicity conservation law (\ref{eq:nl1}), i.e., 
\beqn
\derv{t}\left[\boldsymbol{\Omega}{\bf\cdot}({\bf u}+r\nabla S)\right]
+\nabla{\bf\cdot}\left\{{\bf u}\left[\boldsymbol{\Omega}{\bf\cdot}({\bf u}+r\nabla S)\right]+\boldsymbol{\Omega}\left(h-\frac{1}{2}|{\bf u}|^2\right)\right\}
=0, \label{eq:fh0}
\eeqn
 may be obtained by 
applying Noether's theorem (Proposition \ref{prop8.3}) in which the 
fluid relabelling  symmetry generator ${\hat V}^{\bf x}$ is given 
by the equations:
\begin{align}
{\hat V}^{\bf x}=&\frac{\boldsymbol{\Omega}}{\rho}, 
\quad \boldsymbol{\Omega}= \nabla\times {\bf w}, 
\quad V^t=V^{\bf x}=0, 
\quad V^{{\bf x}_0}=-\frac{(\nabla_0\tilde{\lambda}\times\nabla_0\mu)}
{\rho_0}, \nonumber\\
{\bf w}=&{\bf u}-\nabla\phi+r\nabla S\equiv -\tilde{\lambda}\nabla\mu, 
\label{eq:fh1}
\end{align}
and the gauge potentials $\Lambda^\alpha$ ($\alpha=0,1,2,3$) are given by:
\beqn
\Lambda^0=r(\boldsymbol{\Omega}{\bf\cdot}\nabla S), 
\quad \Lambda^j=r(\boldsymbol{\Omega}{\bf\cdot}\nabla S) u^j. \label{eq:fh2}
\eeqn
\end{proposition}
\begin{proof}
Because $\rho {\hat V}^{\bf x}=\boldsymbol{\Omega}=\nabla\times{\bf w}$ 
it follows that $\nabla{\bf\cdot}(\rho\hat{V}^{\bf x})=0$ in (\ref{eq:3.25n}). 
Also the one form $\alpha={\bf w}{\bf\cdot} d{\bf x}$ is Lie dragged with 
the flow 
and $\hat{V}^{\bf x}= \nabla\times{\bf w}/\rho\equiv\boldsymbol{\Omega}/\rho$
is an invariant advected vector field, i.e., it satisfies the equation:
\beqn
\frac{d\hat{V}^{\bf x}}{dt}-\hat{V}^{\bf x}{\bf\cdot}\nabla{\bf u}
\equiv \deriv{\hat{V}^{\bf x}}{t} +\left[{\bf u},\hat{V}^{\bf x}\right]=0. 
\label{eq:fh3}
\eeqn
The left hand side of (\ref{eq:3.25n}) reduces to:
\beqn
\rho T\hat{V}^{\bf x}{\bf\cdot}\nabla S
= T\boldsymbol{\Omega}{\bf\cdot}\nabla S.  \label{eq:fh4}
\eeqn
The gauge potential divergence term on the right hand side of (\ref{eq:3.25n}) 
reduces to 
\begin{align}
-\nabla_\alpha\Lambda^\alpha
=&-\left(\deriv{\Lambda^0}{t}+\nabla{\bf\cdot}\boldsymbol{\Lambda}\right)
=-\left(\derv{t}\left(r\boldsymbol{\Omega}{\bf\cdot}\nabla S\right)
+\nabla{\bf\cdot}\left({\bf u} (r\boldsymbol{\Omega}{\bf\cdot}\nabla S)
\right)\right)\nonumber\\
=&-\left[r\rho\frac{d}{dt}
\left(\frac{\boldsymbol{\Omega}{\bf\cdot}\nabla S}{\rho}\right) 
+\frac{\boldsymbol{\Omega}{\bf\cdot}\nabla S}{\rho} \rho 
\frac{dr}{dt}\right]
=-\boldsymbol{\Omega}{\bf\cdot}\nabla S \frac{dr}{dt} 
=T\boldsymbol{\Omega}{\bf\cdot}\nabla S, \label{eq:fh5}
\end{align}
which is the same as the left handside (\ref{eq:fh4}). Thus the condition
(\ref{eq:3.25n}) for 
a divergence, relabelling symmetry of the action is satisfied. Using 
(\ref{eq:fh1}) and (\ref{eq:fh2}) 
in the Noether's theorem (proposition \ref{prop8.3})
gives the nonlocal fluid helicity conservation law  
(\ref{eq:fh0}). 
\end{proof}

\begin{proposition}\label{prop8.6}
The nonlocal cross helicity conservation law (\ref{eq:gch1}): 
\beqn
\derv{t}\left[{\bf B}{\bf\cdot}\left({\bf u}+r\nabla S\right)\right]
+\nabla{\bf\cdot}\bigl\{ {\bf u}
\left[{\bf B}{\bf\cdot}
\left({\bf u}+r\nabla S\right)
\right]
+\left(h-\frac{1}{2}|{\bf u}|^2\right) {\bf B}\bigr\}=0,
\label{eq:fh5a}
\eeqn
 is obtained by using Noether's theorem 
(proposition \ref{prop8.3}) 
 in which the fluid relabelling 
symmetry generator $\hat{V}^{\bf x}$ has the form:
\beqn
\hat{V}^{\bf x}=\frac{\bf B}{\rho}\equiv {\bf b}, \quad V^{\bf x}=V^t=0, 
\quad V^{{\bf x}_0}=-\frac{{\bf B}_0}{\rho_0}, \label{eq:chn1}
\eeqn
and the gauge potentials $\Lambda^\alpha$ ($\alpha=0,1,2,3$) are:
\beqn
\Lambda^0=r {\bf B}{\bf\cdot}\nabla S, \quad 
\Lambda^i= u^i r {\bf B}{\bf\cdot}\nabla S. \label{eq:chn2}
\eeqn
\end{proposition}

\begin{proof}
The vector field $\hat{V}^{\bf x}={\bf b}={\bf B}/\rho$ is Lie dragged with 
the flow, and satisfies the equation:
\beqn
\frac{d{\bf b}}{dt}-{\bf b}{\bf\cdot}\nabla {\bf u}\equiv \deriv{\bf b}{t}
+\left[ {\bf u},{\bf b}\right]=0. \label{eq:chn3}
\eeqn
Also note that $\nabla{\bf\cdot}\left(\rho \hat{V}^{\bf x}\right)
=\nabla{\bf\cdot}{\bf B}=0$ 
(Gauss's law). Thus, the left hand side of (\ref{eq:3.25n}) reduces to:
\beqn
\rho T\hat{V}^{\bf x}{\bf\cdot}\nabla S=T {\bf B}{\bf\cdot}\nabla S. 
\label{eq:chn4}
\eeqn
The  divergence term on the right hand side of (\ref{eq:3.25n}) reduces to:
\begin{align}
-\nabla_\alpha\Lambda^\alpha
=&-\left(\derv{t} 
\left( r{\bf B}{\bf\cdot}\nabla S\right)
+\nabla{\bf\cdot}\left({\bf u} 
(r{\bf B}{\bf\cdot}\nabla S)\right)\right)
\equiv&-\left({\bf B}{\bf\cdot}\nabla S\right) \frac{dr}{dt}
=T\left({\bf B}{\bf\cdot} \nabla S\right). \label{eq:chn5}
\end{align}
Using (\ref{eq:chn3})-(\ref{eq:chn5}) in (\ref{eq:3.25n}) shows that 
the Lie invariance condition (\ref{eq:3.25n}) 
is satisfied. Use of (\ref{eq:chn1})-(\ref{eq:chn2}) 
in Noether's theorem (Proposition \ref{prop8.3}) gives the 
nonlocal cross helicity conservation law (\ref{eq:gch1}) or 
(\ref{eq:fh5a}).
\end{proof}

\begin{proposition}\label{prop8.7}
The divergence symmetry condition (\ref{eq:3.25n}) has solutions:
\begin{align}
&\hat{V}^{\bf x}={\bf u},\nonumber\\
&\Lambda^0=-\left(\frac{1}{2}\rho |{\bf u}|^2-\varepsilon(\rho,S)
-\rho\Phi({\bf x})-\frac{B^2}{2\mu_0}\right) -\rho f({\bf x}_0), \nonumber\\
&\Lambda^i=-\rho u^i f({\bf x}_0), \label{eq:3.30n}
\end{align}
where $f({\bf x}_0)$ is an arbitrary function of ${\bf x}_0$. The 
gauge potential $\Lambda^0=-L-\rho f({\bf x}_0)$ 
where $L$ is the Eulerian MHD Lagrangian density, including an external 
gravitational potential $\Phi({\bf x})$. The conservation 
laws associated with the solutions (\ref{eq:3.30n}) are the MHD energy 
conservation equation:
\begin{equation}
\derv{t}\left(\frac{1}{2}\rho |{\bf u}|^2+\varepsilon(\rho,S)+\rho\Phi({\bf x})
+\frac{B^2}{2\mu_0}\right) 
+\nabla{\bf\cdot}\left[\rho {\bf u}\left(\frac{1}{2}|{\bf u}|^2+h+\Phi\right)
+\frac{{\bf E}\times{\bf B}}{\mu_0}\right]=0, \label{eq:3.31n}
\end{equation}
and the conservation law:
\begin{equation}
\derv{t}[\rho f({\bf x}_0)]+\nabla{\bf\cdot}[\rho{\bf u}f({\bf x}_0)]=0\quad
\hbox{or}\quad \left(\derv{t}+{\bf u}{\bf\cdot}\nabla\right) f({\bf x}_0)=0. 
\label{eq:3.32n}
\end{equation}
\end{proposition}
\begin{remark}{\bf 1}
The MHD energy conservation equation (\ref{eq:3.31n}) is usually associated 
with the time translation symmetry of the action, 
for which $V^t=1$, $V^{\bf x}=0$, $V^\psi=0$ ($\psi$ is any of the MHD 
physical variables $\rho$, ${\bf u}$, ${\bf B}$ and $S$), and 
$\Lambda^\alpha=0$ ($\alpha=0,1,2,3$). The result (\ref{eq:3.31n}) 
shows that the energy conservation law (\ref{eq:3.31n}) also arises 
as a gauge symmetry of the action associated with the fluid 
relabelling symmetry. 
\end{remark}

\begin{remark}{\bf 2}
The conservation law (\ref{eq:3.32n}) states that an arbitrary 
function $f({\bf x}_0)$ of the Lagrange labels ${\bf x}_0$  is advected 
with the flow. 
Non-trivial 
examples of this conservation law are obtained for: 
\begin{equation}
f_1({\bf x}_0)=\frac{{\bf B\cdot}\nabla S}{\rho}\equiv 
\frac{{\bf B}_0({\bf x}_0){\bf\cdot}\nabla_0 S({\bf x}_0)}
{\rho_0({\bf x}_0)}, 
\quad
f_2({\bf x}_0)=\frac{\bf A\cdot B}{\rho}
\equiv \frac{{\bf A}_0({\bf x}_0){\bf\cdot}
{\bf B}_0({\bf x}_0)}{\rho_0({\bf x}_0)}, \label{eq:3.33n}
\end{equation}
where ${\bf A}$ is chosen so that 
${\bf A\cdot}d{\bf x}={\bf A}_0({\bf x}_0){\bf\cdot}d{\bf x}_0$ 
is advected with the flow. 
\end{remark}

\begin{proof}
To obtain the solutions (\ref{eq:3.30n}) of the Lie determining equations 
(\ref{eq:3.25n}) for a divergence symmetry of the action, we note that with 
$\hat{V}^{\bf x}={\bf u}$, (\ref{eq:3.25n}) reduces to:
\begin{align}
&\nabla{\bf\cdot} (\rho{\bf u})
\left(h+\Phi({\bf x})-\frac{1}{2}|{\bf u}|^2\right) 
+\rho T{\bf u}{\bf\cdot}\nabla S+\rho {\bf u}{\bf\cdot}\deriv{\bf u}{t}
\nonumber\\
&+\frac{\bf B}{\mu_0}{\bf\cdot}
\left[-\nabla\times({\bf u}\times{\bf B}) 
+{\bf u}(\nabla{\bf\cdot}{\bf B})\right]=-\nabla_\alpha \Lambda^\alpha. 
\label{eq:3.34n}
\end{align}
Next we use the identities:
\begin{align}
&T_1=\rho{\bf u}{\bf\cdot}\deriv{\bf u}{t}
-\frac{1}{2}|{\bf u}|^2\nabla{\bf\cdot}(\rho{\bf u})
\equiv \derv{t}\left(\frac{1}{2}\rho |{\bf u}|^2\right)
-\frac{1}{2}|{\bf u}|^2
\left[\deriv{\rho}{t}+\nabla{\bf\cdot}(\rho {\bf u})\right], \nonumber\\
&T_2=\nabla{\bf\cdot}(\rho {\bf u}) h+\rho T{\bf u}{\bf\cdot}\nabla S
=\nabla{\bf\cdot}(\rho {\bf u}h) -{\bf u}{\bf\cdot}\nabla p
\equiv -\deriv{\varepsilon}{t}, 
\nonumber\\
&T_3=\nabla{\bf\cdot}(\rho {\bf u}) \Phi({\bf x})\equiv -\deriv{\rho}{t} \Phi
=-\derv{t}[\rho \Phi({\bf x})], \nonumber\\
&T_4=\frac{\bf B}{\mu_0}{\bf\cdot}\left[-\nabla\times({\bf u}\times{\bf B}) 
+{\bf u}(\nabla{\bf\cdot}{\bf B})\right]\equiv 
-\derv{t}\left(\frac{B^2}{2\mu_0}\right). \label{eq:3.35n}
\end{align}
In (\ref{eq:3.35n}) use of the mass continuity equation (\ref{eq:2.1}) 
gives $T_1=\partial((1/2)\rho |{\bf u}|^2)/\partial t$. 
The term $T_2$ in (\ref{eq:3.35n})  reduces to 
$-\partial\varepsilon/\partial t$, where we have used the internal 
energy evolution equation for the gas:
\beqn
\deriv{\varepsilon}{t}+\nabla{\bf\cdot}(\rho {\bf u}h)
={\bf u}{\bf\cdot}\nabla p,  \label{eq:3.36n}
\eeqn
where $\varepsilon=\varepsilon(\rho,S)$.  
 The expression $T_4$ in (\ref{eq:3.35n}), using Faraday's equation 
reduces to $-\partial(B^2/2\mu_0)/\partial t$. This result 
is Poynting's theorem. 

Using the results (\ref{eq:3.35n}),  (\ref{eq:3.34n}) reduces to:
\beqn
\derv{t}\left(\frac{1}{2}\rho |{\bf u}|^2-\varepsilon(\rho,S)-\rho\Phi({\bf x})
-\frac{B^2}{2\mu_0}\right)=-\left(\deriv{\Lambda^0}{t}
+\deriv{\Lambda^i}{x^i}\right). 
\label{eq:3.37n}
\eeqn
Equation (\ref{eq:3.37n}) has solutions of the form (\ref{eq:3.30n}). 
 
The total energy conservation law (\ref{eq:3.31n}) and the Lagrangian 
advection conservation law (\ref{eq:3.32n}) now follow by using 
the symmetry results (\ref{eq:3.30n}) in Noether's theorem 
(proposition 6.3). 
 From (\ref{eq:3.17n})-(\ref{eq:3.18n}) we find:
\begin{align}
&F^0=\left(\frac{1}{2}\rho |{\bf u}|^2+\varepsilon+\rho\Phi
+\frac{B^2}{2\mu_0}\right)
+\rho f({\bf x}_0), \label{eq:3.38n}\\
&{\bf F}=\biggl[\rho{\bf u}
\left(\frac{1}{2}\rho |{\bf u}|^2
+h+\Phi\right)+\frac{{\bf E}\times{\bf B}}{\mu_0}\biggr] 
+\rho {\bf u} f({\bf x}_0), \label{eq:3.39n}
\end{align}
where ${\bf F}=(F^1,F^2,F^3)$ is the spatial flux and 
${\bf E}=-({\bf u}\times{\bf B})$ is the electric field. 
The MHD energy conservation law (\ref{eq:3.31n}) is obtained by 
setting $f({\bf x}_0)=0$ in (\ref{eq:3.38n})-(\ref{eq:3.39n}) 
and using (\ref{eq:3.38n})-(\ref{eq:3.39n}) for 
$F^0$ and ${\bf F}$ in (\ref{eq:3.16n}). 
Similarly, the conservation law (\ref{eq:3.32n}) for $f({\bf x}_0)$ 
is obtained by using (\ref{eq:3.16n}).  
This completes the proof. 
\end{proof}

\section{Euler-Poincar\'e Equation Approach}
Our analysis in this section is based in part, on the analysis of 
Holm et al. (1998) and Cotter and Holm (2012). 
 In action principles in MHD and gas dynamics, it is useful
to use both Lagrangian and Eulerian variations. The Euler-Poincar\'e approach
uses Eulerian variations in which ${\bf x}$ is held constant. In 
Section 7.1 we derive the MHD Euler-Poincar\'e equation or 
Eulerian momentum equation for MHD (see also Holm et al. (1998)
for a similar approach). In Section 7.2, we give an analysis 
of Noether's second theorem for MHD and fluid relabelling symmetries 
which is similar to the analysis of Cotter and Holm (2012). 
The results from Noether's second theorem using the Euler-Poincar\'e 
approach  
overlap with the more classical 
physics approach in Section 6. However, there are some subtle 
issues in Noether's second theorem that arise in this 
section, which were not addressed in Section 6. 

 The solution of $d{\bf x}/dt={\bf u}({\bf x},t)$
with ${\bf x}={\bf x}_0$ at $t=0$ is written as
${\bf x}=g{\bf x}_0={\bf X}({\bf x}_0,t)$. 
The inverse map ${\bf x}_0=g^{-1} {\bf x}$ 
defines ${\bf x}_0={\bf x}_0({\bf x},t)$.
The Lagrange label ${\bf x}_0$ is advected with the flow:
\beqn
\left(\derv{t}+{\bf u}{\bf\cdot}\nabla\right){\bf x}_0=\left(\derv{t}+
{\cal L}_{\bf u}\right) {\bf x}_0=0, \label{eq:ep1}
\eeqn

Write ${\bf x}=g{\bf x}_0$, ${\bf x}_0=g^{-1}{\bf x}$.
Notice that $\dot{\bf x}_0=(g^{-1})^{\bf\dot{}}{\bf x}=-g^{-1}g g^{-1}{\bf x}_0$
(use $g^{-1}g=e$ where $e$ is the identity).
 Here $\dot{\bf x}_0=\partial {\bf x}_0/\partial t$
where ${\bf x}$ is held constant. Thus,
\beqn
\dot{\bf x}_0=-g^{-1}\dot{g}g^{-1} g{\bf x}_0
=-g^{-1}\dot{g}{\bf x}_0=-{\cal L}_{\bf u}{\bf x}_0. \label{eq:ep2}
\eeqn
We identify
\beqn
\boldsymbol{\xi}={\cal L}_{\bf u}={\bf u}{\bf\cdot}\nabla\equiv g^{-1}\dot{g},
 \label{eq:ep3}
\eeqn
with the fluid velocity ${\bf u}$.
Note $\boldsymbol{\xi}=g^{-1} \dot{g}$ is left
invariant vector field.
 Similarly, for a geometrical object Lie dragged with the flow:
\beqn
\left(\derv{t}+{\cal L}_{\bf u}\right) a=0. \label{eq:ep4}
\eeqn
Let $a_0=ga$ then $a=g^{-1}a_0$ and
\beqn
\delta a=\delta\left(g^{-1}\right)a_0
=-g^{-1}\delta g g^{-1}a_0=-g^{-1}\delta g\ a=-{\cal L}_\eta (a). \label{eq:ep5}
\eeqn
We write
\beqn
\eta=g^{-1}\delta g. \label{eq:ep6}
\eeqn
as the vector field associated with the variations. Note $\eta$ is a left
invariant vector field (i.e. $(hg)^{-1}\delta(hg)=g^{-1}\delta g$, 
assuming that $\delta h=0$). 

To compute $\delta\xi$ where $\xi=g^{-1}\dot{g}$ we note:
\beqn
\delta\xi=\delta g^{-1} \dot{g}+g^{-1}\delta\dot{g}
=-(g^{-1}\delta g g^{-1})\dot{g}+g^{-1}\delta\dot{g}, \label{eq:ep7}
\eeqn
 which gives:
\beqn
\delta\xi=-\eta\xi+g^{-1}\delta\dot{g}, \label{eq:ep8}
\eeqn

Similarly, for $\eta=g^{-1}\delta g$ we find
\beqn
\dot{\eta}=(g^{-1})^{\bf\dot{}}\delta g+g^{-1}\delta\dot{g}
=-g^{-1}\dot{g}g^{-1}\delta g+g^{-1}\delta\dot{g}, \label{eq:ep9}
\eeqn
which gives:
\beqn
\dot{\eta}=-\xi\eta+g^{-1}\delta\dot{g}. \label{eq:ep10}
\eeqn
Subtract (\ref{eq:ep10}) from (\ref{eq:ep8}) gives:
\beqn
\delta\xi=\dot{\eta}+\xi\eta-\eta\xi\equiv \dot{\eta}+[\xi,\eta]_L.
\label{eq:ep11}
\eeqn
where $[\xi,\eta]_L=ad_{\xi}(\eta)_L$ is the left Lie bracket. The right
Lie bracket $[\xi,\eta]_R=-[\xi,\eta]_L$.

\subsection{The Euler-Poincar\'e equation}

Consider the  variational principle
 (Holm et al. (1998), Cotter and Holm (2012))
in which the action: 
\beqn
J=\int \ell ({\bf u},a)\ d^3x\ dt, \label{eq:action1}
\eeqn
is stationary, i.e.
\beqn
\delta J=\int \left(\frac{\delta \ell}{\delta {\bf u}}
{\bf\cdot}\delta {\bf u}
+\frac{\delta\ell}{\delta a}\delta a\right)\ d^3x\ dt
\equiv \int\left\langle \frac{\delta \ell}{\delta {\bf u}}, \delta {\bf u}
\right\rangle
+\left\langle \frac{\delta\ell}{\delta a}, \delta a\right\rangle\ dt=0.
\label{eq:action2}
\eeqn

However from (\ref{eq:ep11}) with $\xi={\bf u}$, and (\ref{eq:ep5}),   
\beqn
\delta {\bf u}=\dot{\eta}+[{\bf u},\eta],\quad \delta a=-{\cal L}_\eta(a).
\label{eq:action3}
\eeqn
Thus
\beqn
\delta J=\int \left\langle
\frac{\delta \ell}{\delta {\bf u}}, \dot{\eta}
+[{\bf u},\eta]\right\rangle
+\left\langle \frac{\delta\ell}{\delta a},
-{\cal L}_\eta(a)\right\rangle\ dt. \label{eq:action4}
\eeqn
\ Integrate (\ref{eq:action4}) by parts, and use $ad_{\bf u}(\eta) 
=[{\bf u},\eta]$ to obtain:
\beqnar
\delta J=&&\int\ \biggl\{
\left(\derv{t}\left\langle
\frac{\delta \ell}{\delta {\bf u}},
\eta\right\rangle
-\left\langle\ \eta,\derv{t}\left(\frac{\delta \ell}{\delta {\bf u}}\right)
\right\rangle\right)
+\left\langle \frac{\delta \ell}{\delta{\bf u}},ad_{\bf u}(\eta)
\right\rangle\nonumber\\\
&&-\left\langle\frac{\delta\ell}{\delta a}, {\cal L}_{\eta}(a)\right\rangle
\biggr\}\  dt.
\label{eq:action5}
\eeqnar
for $\delta J$.

In the further analysis of (\ref{eq:action5}) it is useful to introduce 
the diamond operator. The diamond operator $\diamond$  
 allows one to take the adjoint of the 
$\langle \delta\ell/\delta {\bf u},ad_{\bf u}(\eta)\rangle$ 
term in (\ref{eq:action5}) and thereby isolate 
its $\eta$ component, by using the formula
\beqn
\left\langle\frac{\delta\ell}{\delta a}\diamond a,\eta\right\rangle=-
\left\langle\frac{\delta \ell}{\delta {a}},{\cal L}_{\eta}(a)\right\rangle.
\label{eq:action5a}
\eeqn
A more formal definition of the diamond operator is given below. 

\begin{definition}
The diamond operator $\diamond$ is  minus the dual of the Lie 
derivative, with respect to the pairing induced by the variational derivative 
$p=\delta\ell/\delta q$, namely:
\beqn
\langle p\diamond q,\xi\rangle=\langle p,-{\cal L}_{\xi}(q)\rangle. 
\label{eq:action5b}
\eeqn
\end{definition}
Using (\ref{eq:action5a}) and the definition of $ad^*_{\bf u}$:
\beqn
 \left\langle\ ad_{\bf u}^*\left(\frac{\delta \ell}{\delta{\bf u}}\right)
,\eta\right\rangle
=\left\langle \frac{\delta \ell}{\delta{\bf u}},ad_{\bf u}(\eta)
\right\rangle, \label{eq:action6}
\eeqn
in (\ref{eq:action5}) where $\diamond$ is the diamond operator (this involves
 integration by parts, and dropping surface terms). We obtain:
\beqn
\delta J=\int\left\langle\eta,-\derv{t}
\left(\frac{\delta \ell}{\delta {\bf u}}\right)+ad_{\bf u}^*
\left(\frac{\delta \ell}{\delta{\bf u}}\right)+
\frac{\delta\ell}{\delta a}\diamond a\right\rangle\ dt
+\left[\left\langle\frac{\delta \ell}{\delta {\bf u}},
\eta\right\rangle\right]_{t_0}^{t_1}. \label{eq:action7}
\eeqn
Assuming the surface term vanishes in (\ref{eq:action7}), and $\eta$
is arbitrary, then $\delta J=0$ implies the Euler-Poincar\'e equation:
\beqn
\derv{t}\left(\frac{\delta \ell}{\delta {\bf u}}\right)
+ad_{\bf u}^*\left(\frac{\delta \ell}{\delta {\bf u}}\right)_R
=\frac{\delta\ell}{\delta a}\diamond a, \label{eq:action8}
\eeqn
where
\beqn
ad_{\bf u}^*\left(\frac{\delta \ell}{\delta {\bf u}}\right)_R
=-ad_{\bf u}^*\left(\frac{\delta \ell}{\delta {\bf u}}\right)_L.
\label{eq:action9}
\eeqn
Here, (\ref{eq:action8}) is the Euler Poincar\'e equation 
for the variational principle $\delta J=0$ (Holm  
et al. (1998)).
In (\ref{eq:action8}), $d/dt\equiv \partial/\partial t$ 
keeping ${\bf x}$ constant.
Below, we show that: 
\beqn
ad_{\bf u}^*\left(\frac{\delta \ell}{\delta {\bf u}}\right)_R=
\nabla{\bf\cdot}\left({\bf u}\otimes \frac{\delta \ell}{\delta {\bf u}}\right)
+(\nabla{\bf u})^T{\bf\cdot} \left(\frac{\delta \ell}{\delta {\bf u}}\right).
\label{eq:action10}
\eeqn

\begin{proof}
 To prove (\ref{eq:action10}) let
${\bf m}=\delta\ell/\delta {\bf u}$.   We obtain:
\beqnar
&&\left\langle\eta,ad_{\bf u}^*({\bf m})_R\right\rangle
=\left\langle ad_{\bf u}(\eta), {\bf m}\right\rangle
=\left\langle -[{\bf u},\eta]_L,{\bf m}\right\rangle\nonumber\\
&&=-\int\left[\left({\bf u\cdot}\nabla\boldsymbol{\eta}
-\boldsymbol{\eta}{\bf\cdot}\nabla{\bf u}\right)
\nabla\right]{\mathrel{\lrcorner}}
{\bf m\cdot}d{\bf x}\ d^3x\nonumber\\
&&=\int-\nabla{\bf\cdot} ({\bf u} ({\bf m\cdot}{\boldsymbol{\eta}})
)+\boldsymbol{\eta}{\bf\cdot}\left(\nabla{\bf\cdot}({\bf u}\otimes{\bf m})
+(\nabla{\bf u})^T{\bf\cdot m}\right)\ d^3x, \nonumber\\
&&=\left\langle \boldsymbol{\eta},\nabla{\bf\cdot}({\bf u}\otimes{\bf m})
+(\nabla{\bf u})^T{\bf\cdot m}\right\rangle, \label{eq:action11}
\eeqnar
where we dropped the surface term. This proves (\ref{eq:action10}).
\end{proof}

It can be shown that:
\beqn
{\cal L}_{\bf u}({\bf m\cdot}d{\bf x}\otimes dV)
=\left(\nabla{\bf\cdot}({\bf u}\otimes{\bf m})
+(\nabla{\bf u})^T{\bf\cdot m}\right){\bf\cdot}d{\bf x}\otimes dV.
\label{eq:action12}
\eeqn

For MHD the Lagrange density $\ell$ is given by:
\beqn
\ell=\frac{1}{2}\rho u^2-\varepsilon(\rho,S)-\frac{B^2}{2\mu_0}.
\label{eq:action13}
\eeqn

We now determine the different terms in the Euler-Poincar\'e equation 
(\ref{eq:action8}). 

From (\ref{eq:action2}), the variation of the action
$\delta J=\delta J_{\bf u}+\delta J_a$  where:
\beqnar
&&\delta J_{\bf u}=\int \frac{\delta\ell}{\delta{\bf u}}
{\bf\cdot}\delta {\bf u} \ d^3x\ dt,
\nonumber\\
&&\delta J_a=\int \left(\frac{\delta\ell}{\delta\rho}\delta\rho
+\frac{\delta\ell}{\delta S}\delta S
+\frac{\delta\ell}{\delta{\bf B}}{\bf\cdot}\delta {\bf B}\right)\ d^3x\ dt.
\label{eq:action14}
\eeqnar
 From (\ref{eq:action13}) we obtain:
\beqnar
&&\frac{\delta\ell}{\delta\rho}=\frac{1}{2}u^2-\varepsilon_\rho=\frac{1}{2}u^2-h,
\quad
\frac{\delta\ell}{\delta{\bf u}}\equiv {\bf m}=\rho{\bf u}, \nonumber\\
&&\frac{\delta\ell}{\delta S}=-\varepsilon_S=-\rho T, \quad
\frac{\delta\ell}{\delta{\bf B}}=-\frac{\bf B}{\mu_0}, \label{eq:action15}
\eeqnar
where $T$ is the temperature  and $h$ is the enthalpy of the gas.

Using the formulae:
\beqnar
&&\delta\left(\rho d^3x\right)=-{\cal L}_{\bf u}\left(\rho d^3 x\right)
=-\nabla{\bf\cdot}(\rho {\bf u})\ d^3x, \nonumber\\
&&\delta S=-{\cal L}_{\bf u} (S)=
-{\bf u}{\bf\cdot}\nabla S, \nonumber\\
&&\delta ({\bf B}{\bf\cdot}d{\bf S})
=-{\cal L}_{\bf u}({\bf B}{\bf\cdot}d{\bf S})\nonumber\\
&&=[\nabla\times({\bf u}\times{\bf B})-{\bf u}(\nabla{\bf\cdot B})]
{\bf\cdot}d{\bf S},
  \label{eq:action16}
\eeqnar
we obtain:
\beqnar
&&\delta\rho=-\nabla{\bf\cdot}(\rho{\bf u}),\quad
\delta S=-{\bf u}{\bf\cdot}\nabla S, \nonumber\\
&&\delta{\bf B}
=[\nabla\times({\bf u}\times{\bf B})-{\bf u}(\nabla{\bf\cdot B})].
\label{eq:action17}
\eeqnar
 Note that $\delta\rho$, $\delta S$ and $\delta {\bf B}$ are
  Eulerian variations in which
$\Delta {\bf x}^i=-x_{ij} \delta x_0^j$ is replaced by $u^i$,
where $\Delta{\bf x}$ is the Lagrangian variation of ${\bf x}$,
and $x_{ij}=\partial x^i/\partial x_0^j$
(e.g. Webb et al. (2005a,b), Newcomb (1962)).
Using $\delta\ell/\delta{\bf u}=\rho{\bf u}={\bf m}$
 in (\ref{eq:action10}) gives:
\beqn
ad_{\bf u}^*\left(\frac{\delta\ell}{\delta{\bf u}}\right)_R
=\nabla{\bf\cdot}(\rho {\bf u}\otimes{\bf u})
+\rho\nabla\left(\frac{1}{2}|{\bf u}|^2\right), \label{eq:action18}
\eeqn
for the advected term on the left hand side  of the Euler-Poincar\'e 
equation  (\ref{eq:action8}).

Next we find the 
 $(\delta\ell/\delta a)\diamond a$ term on
right hand side  of (\ref{eq:action8}). We obtain:
\beqnar
\frac{\delta\ell}{\delta a}\delta a
=&&\frac{\delta\ell}{\delta \rho} \delta\rho
+\frac{\delta\ell}{\delta S} \delta S
+\frac{\delta\ell}{\delta {\bf B}}{\bf\cdot} \delta {\bf B}\nonumber\\
=&&\frac{\delta\ell}{\delta \rho}\left(-\nabla{\bf\cdot}(\rho{\bf u})\right)
+\frac{\delta\ell}{\delta S} \left(-{\bf u}{\bf\cdot}\nabla S\right)\nonumber\\
&&+\frac{\delta\ell}{\delta {\bf B}}
{\bf\cdot}\left[\nabla\times({\bf u}\times{\bf B})
-{\bf u}\nabla{\bf\cdot B}\right].
\label{eq:action19}
\eeqnar
Thus
\beqnar
&&\frac{\delta\ell}{\delta a}\delta a
=-\nabla{\bf\cdot}\left(\rho{\bf u}\frac{\delta\ell}{\delta \rho}\right)
+\nabla{\bf\cdot}
\left[({\bf u}\times{\bf B})\times \frac{\delta\ell}{\delta {\bf B}}\right]
\nonumber\\
&&+{\bf u}{\bf\cdot}
\biggl\{\rho\nabla\left(\frac{\delta\ell}{\delta \rho}\right)
-\frac{\delta\ell}{\delta S}\nabla S
+{\bf B}
\times\left(\nabla\times\left(\frac{\delta\ell}{\delta {\bf B}}\right)
\right)
-\frac{\delta\ell}{\delta {\bf B}}\nabla{\bf\cdot}{\bf B}\biggr\}.
\label{eq:action20}
\eeqnar
From (\ref{eq:action20}) we find:
\beqn
\frac{\delta\ell}{\delta a}\diamond  a
=\rho\nabla\left(\frac{\delta\ell}{\delta \rho}\right)
-\frac{\delta\ell}{\delta S}\nabla S
+{\bf B}
\times\left(\nabla\times\left(\frac{\delta\ell}{\delta {\bf B}}\right)
\right)
-\frac{\delta\ell}{\delta {\bf B}}\nabla{\bf\cdot}{\bf B}. \label{eq:action21}
\eeqn
 Integrate (\ref{eq:action20}) over $d^3x$
over the volume, $V$, drop surface terms, and set $\eta\to {\bf u}$
in (\ref{eq:action7}) gives the result (\ref{eq:action21}) 
for $\delta\ell/\delta a\diamond a$. 

Using the first law of thermodynamics in the form:
$T\nabla S-\nabla h=-\nabla p/\rho$ and
the expressions  (\ref{eq:action15}) for
$\delta\ell/\delta\rho$, $\delta\ell/\delta S$, $\delta\ell/\delta {\bf B}$
in (\ref{eq:action21}) gives:
\beqn
\frac{\delta\ell}{\delta a}\diamond  a
=\left(-\nabla p+{\bf J}\times{\bf B}
+\frac{\bf B}{\mu_0}\nabla{\bf\cdot}{\bf B}\right)
+\rho\nabla\left(\frac{1}{2}
|{\bf u}|^2\right). \label{eq:action22}
\eeqn

 Using  $ad^*_{\bf u}(\delta\ell/\delta {\bf u})_R$ from  
 (\ref{eq:action18}) and  $\delta\ell/\delta a\diamond a$ from
 (\ref{eq:action22}) in the Euler Poincar\'e  equation
(\ref{eq:action8}) gives the MHD momentum equation in the form:
\beqn
\derv{t}\left(\rho{\bf u}\right)+\nabla{\bf\cdot}
\left(\rho{\bf u}\otimes{\bf u}\right)
=-\nabla p+{\bf J}\times{\bf B}
+\frac{\bf B}{\mu_0}\nabla{\bf\cdot}{\bf B}.
 \label{eq:action23}
\eeqn
The momentum equation (\ref{eq:action23}) can also 
be written in the conservative form:
\beqn
\derv{t}\left(\rho{\bf u}\right)
+\nabla{\bf\cdot}\left(\rho{\bf u}\otimes{\bf u}+\left(
p+\frac{B^2}{2\mu_0}\right){\sf I}
-\frac{{\bf B}\otimes{\bf B}}{\mu_0}\right)
=0,
\label{eq:action24}
\eeqn
where the magnetic terms involve the Maxwell stress energy tensor.
The above derivation of the Euler-Poincar\'e equation is essentially
that of Holm et al. (1998). It is also discussed by Cotter and Holm (2012)
in their analysis of symmetries and conservation laws associated with
advection of physical quantities i.e., the Tur and Yanovsky (1993)
conservation laws.

\subsection{Noether's second theorem}

Consider the application of the above ideas to obtain a version 
of Noether's second theorem associated with the symmetries 
$\boldsymbol{\eta}$. In the derivation 
of Noether's theorem, it is useful to keep track of all the 
surface or divergence terms that arise when integrating by parts. These terms 
are assumed to vanish in the derivation of the Euler-Poincar\'e 
equation (\ref{eq:action23}) or (\ref{eq:action24}). The variation of the 
action $\delta J$ is again given by (\ref{eq:action2}), which reduces 
to the result (\ref{eq:action4}), i.e.
\beqn
\delta J=\int \left\langle
\frac{\delta \ell}{\delta {\bf u}}, \dot{\eta}
+[{\bf u},\eta]\right\rangle
+\left\langle \frac{\delta\ell}{\delta a},
-{\cal L}_\eta(a)\right\rangle\ dt\equiv \delta J_u+\delta J_a, 
\label{eq:noeth1}
\eeqn
where $\delta J_u$ and $\delta J_a$ are given by (\ref{eq:action14}). 
Using integration by parts, the first term $\delta J_u$ in (\ref{eq:noeth1}) 
reduces to:
\beqnar
\delta J_u&=&-\int\left\langle \eta,\derv{t}
\left(\frac{\delta\ell}{\delta{\bf u}}\right)
+ad_{\bf u}^*\left(\frac{\delta\ell}{\delta{\bf u}}\right)_R\right\rangle 
dt\nonumber\\ 
&\ &+\int\derv{t}
\left(\boldsymbol{\eta}{\bf\cdot}\frac{\delta\ell}{\delta{\bf u}}\right)
+\nabla{\bf\cdot}\left[\left(\boldsymbol{\eta}{\bf\cdot}\frac{\delta\ell}{\delta{\bf u}}\right) {\bf u}\right]\ d^3x dt. 
\label{eq:noeth2}
\eeqnar

The variations of the $a$ variables is given by (\ref{eq:ep5}), i.e. 
$\delta a=-{\cal L}_\eta(a)$. Thus, we compute the variations 
$\delta(\rho d^3x)$, $\delta S$ and $\delta ({\bf B}{\bf\cdot}d{\bf S})$ 
as in (\ref{eq:action16}) 
but with ${\bf u}$ replaced by $\boldsymbol{\eta}$. The net 
result from (\ref{eq:action17}) is:
\beqnar
&&\delta\rho=-\nabla{\bf\cdot}(\rho\boldsymbol{\eta}),\quad
\delta S=-\boldsymbol{\eta}{\bf\cdot}\nabla S, \nonumber\\
&&\delta{\bf B}
=[\nabla\times(\boldsymbol{\eta}\times{\bf B})
-\boldsymbol{\eta}(\nabla{\bf\cdot B})].
\label{eq:noeth3}
\eeqnar
Using the results (\ref{eq:action15}) and (\ref{eq:noeth3}) we obtain
equation (\ref{eq:action19}) 
but with ${\bf u}$ replaced by $\boldsymbol{\eta}$. 
The net upshot is the result (\ref{eq:action20}) but 
with ${\bf u}$ replaced by 
$\boldsymbol{\eta}$, i.e.,
\beqnar
&&\frac{\delta\ell}{\delta a}\delta a
=-\nabla{\bf\cdot}
\left(\rho\boldsymbol{\eta}\frac{\delta\ell}{\delta \rho}\right)
+\nabla{\bf\cdot}
\left[(\boldsymbol{\eta}\times{\bf B})\times 
\frac{\delta\ell}{\delta {\bf B}}\right]
\nonumber\\
&&+\boldsymbol{\eta}{\bf\cdot}
\biggl\{\rho\nabla\left(\frac{\delta\ell}{\delta \rho}\right)
-\frac{\delta\ell}{\delta S}\nabla S
+{\bf B}
\times\left(\nabla\times\left(\frac{\delta\ell}{\delta {\bf B}}\right)
\right)
-\frac{\delta\ell}{\delta {\bf B}}\nabla{\bf\cdot}{\bf B}\biggr\}.
\label{eq:noeth4}
\eeqnar
Using (\ref{eq:noeth4}) we obtain:
\beqnar
\delta J_a&=&\int\frac{\delta\ell}{\delta a}\delta a\ d^3 x\ dt\nonumber\\
&\ &\int\left\langle\boldsymbol{\eta},\frac{\delta\ell}{\delta a}
\diamond a\right\rangle\ dt
+\int\nabla{\bf\cdot}
\left(-\rho\boldsymbol{\eta}\frac{\delta\ell}{\delta \rho}
+(\boldsymbol{\eta}\times{\bf B})\times
\frac{\delta\ell}{\delta {\bf B}}\right) d^3x\ dt, 
\label{eq:noeth5}
\eeqnar 
where $\delta\ell/\delta a\diamond a$ is given by (\ref{eq:action21}), 
or the coefficient of $\boldsymbol{\eta}$ in (\ref{eq:noeth4}). 
Adding (\ref{eq:noeth2}) and (\ref{eq:noeth5}) for $\delta J_u$ 
and $\delta J_a$ we obtain:
\begin{align}
&\delta J=\delta J_u+\delta J_a
=-\int\left\langle\boldsymbol{\eta}, 
\derv{t}\left(\frac{\delta \ell}{\delta {\bf u}}\right)
+ad_{\bf u}^*\left(\frac{\delta \ell}{\delta {\bf u}}\right)_R
-\frac{\delta\ell}{\delta a}\diamond a\right\rangle dt\nonumber\\
&+\int\int\left[\derv{t}\left(\boldsymbol{\eta}{\bf\cdot}
\frac{\delta\ell}{\delta{\bf u}}\right) 
+\nabla{\bf\cdot}\left( 
\boldsymbol{\eta}{\bf\cdot}\frac{\delta\ell}{\delta{\bf u}} {\bf u}
-\rho\boldsymbol{\eta}\frac{\delta\ell}{\delta \rho}
+(\boldsymbol{\eta}\times{\bf B})\times
\frac{\delta\ell}{\delta {\bf B}}\right)\right]d^3xdt.\nonumber\\
& \label{eq:noeth6}
\end{align}
We require $\delta J=0$ in (\ref{eq:noeth6}) in order for $\boldsymbol{\eta}$
to be a variational symmetry of the action. Because there are an 
infinite number of fluid relabeling symmetries $\boldsymbol{\eta}$ one cannot 
automatically assume that the Euler Lagrange equations (\ref{eq:action8}) 
are satisfied. We can write (\ref{eq:noeth6}) in the form:
\beqn
\delta J=\int\left\langle\boldsymbol{\eta},E_{\{{\bf u},a\}} 
\left(\ell\right)\right\rangle\ dt 
+\int\int \left(\deriv{D}{t}
+\nabla{\bf\cdot}{\bf F}\right)\ d^3x\ dt, \label{eq:noeth6a}
\eeqn
where
\beqn
E_{\{{\bf u},a\}}\left(\ell\right)=-\left\{\derv{t}\left(\frac{\delta \ell}{\delta {\bf u}}\right)
+ad_{\bf u}^*\left(\frac{\delta \ell}{\delta {\bf u}}\right)_R
-\frac{\delta\ell}{\delta a}\diamond a\right\}, \label{eq:noeth6b}
\eeqn
is the Euler operator and 
\beqn
D=\boldsymbol{\eta}{\bf\cdot}
\frac{\delta\ell}{\delta{\bf u}}, \quad
{\bf F}=\boldsymbol{\eta}{\bf\cdot}\frac{\delta\ell}{\delta{\bf u}} {\bf u}
-\rho\boldsymbol{\eta}\frac{\delta\ell}{\delta \rho}
+(\boldsymbol{\eta}\times{\bf B})\times
\frac{\delta\ell}{\delta {\bf B}}, \label{eq:noeth6c}
\eeqn
are the density $D$ and flux ${\bf F}$ surface terms. Further analysis 
of (\ref{eq:noeth6a}) involving integration by parts is necessary before 
one can arrive at a conservation law for particular Lie symmetries 
(which involve arbitrary function(s)). In particular, Padhye and Morrison 
(1996a,b) and Padhye (1998) use this procedure to obtain Ertel's 
theorem, from  fluid relabelling symmetries. 

The variational equation (\ref{eq:noeth6a}) can be written in the form:
\beqn
\delta J=\int\left\langle\boldsymbol{\eta},{\bf E}
\left(\ell\right)\right\rangle\ dt
+C(t)+\int\int\nabla{\bf\cdot}{\bf F}\ d^3x\ dt,  \label{eq:noeth6d}
\eeqn
where
\beqnar
C(t)&=&\int\int \deriv{D}{t}
\ d^3x\ dt\equiv 
\left[\left\langle{\boldsymbol{\eta}},
\frac{\delta\ell}{\delta{\bf u}}\right\rangle\right]_{t_0}^t,\nonumber\\
\langle\boldsymbol{\eta},\frac{\delta \ell}{\delta {\bf u}}\rangle 
&=&\int_V\ d^3x\ 
\left(\boldsymbol{\eta}{\bf\cdot}\frac{\delta \ell}{\delta {\bf u}}\right),
 \label{eq:noeth6e}
\eeqnar
and $D$ and ${\bf F}$ are given by (\ref{eq:noeth6c}).  

 Using the formulae 
(\ref{eq:action15}) for $\delta\ell/\delta\rho$,
$\delta\ell/\delta{\bf u}$, $\delta\ell/\delta S$ and
$\delta\ell/\delta {\bf B}$ in (\ref{eq:noeth6c}) gives:
\beqnar
&&D={\hat V}^{\bf x}{\bf\cdot}\rho{\bf u}+\Lambda^0, \nonumber\\
&&{\bf F}={\hat V}^{\bf x}{\bf\cdot}
\left(\rho{\bf u}\otimes{\bf u}+
\left(\varepsilon+p+\frac{B^2}{\mu_0}-\frac{1}{2}\rho |{\bf u}|^2\right)
{\sf I}-\frac{{\bf B}\otimes{\bf B}}{\mu_0}\right)+\boldsymbol{\Lambda}, 
\label{eq:noeth6f}
\eeqnar 
where use the notation:
\beqn
\hat{V}^{\bf x}=\boldsymbol{\eta}. \label{eq:noeth9}
\eeqn
and we have added potentials $\Lambda^0$ and $\boldsymbol{\Lambda}$ 
in (\ref{eq:noeth6f}) to account for the possibility of 
gauge transformations, which agrees with the density and flux 
formulas obtained in Section 6 in (\ref{eq:3.20na})-(\ref{eq:3.20nb}), 
for the conserved density and flux in Noether's theorem for  
 fluid relabelling symmetries and gauge transformations.  
Here ${\hat V}^{\bf x}$ is the canonical symmetry generator for
fluid relabeling symmetries, in which ${\bf x}={\bf x}({\bf x}_0,t)$
is the Lagrangian map, in which the $x^i$ are the dependent variables
and Lagrange labels ${\bf x}_0$ are the independent variables (e.g.
Webb et al. (2005b), Webb and Zank (2007)). From Ibragimov (1985)
and Webb et al. (2005b)
\beqn
\hat{V}^{x^i}=V^{x^i}-V^{x_0^s}D_{x_0^s} x^i\equiv -V^{x_0^s} x_{is},
\label{eq:noeth10}
\eeqn
gives the formula for the canonical symmetry generator ${\hat V}^{\bf x}$ 
in terms of the Lagrange label symmetry generator $V^{x_0^s}$ 
where $x_{is}=\partial x^i/\partial x_0^s$.  


\subsubsection{Fluid relabeling determining equations}  
For fluid relabeling symmetries, Eulerian physical
variables do not change (e.g. Webb and Zank (2007)). Advected
quantities $a$ satisfy:
\beqn
\delta a=-{\cal L}_{\eta} (a)=0, \label{eq:noeth11}
\eeqn
where $\eta$ is the vector field generator of the relabeling symmetry.

The Eulerian fluid velocity ${\bf u}$ does not change under
fluid relabeling symmetry. Thus,
\beqn
\delta{\bf u}=\dot{\eta}+[{\bf u},\eta]=0. \label{eq:noeth12}
\eeqn
 (\ref{eq:noeth12}) is condition for the vector field
$\eta$ to be Lie dragged by the fluid, i.e. $d\eta/dt=0$ moving with
the flow. 

The conditions (\ref{eq:noeth11}) are equivalent in the case of MHD
of setting $\delta\rho$, $\delta S$ and $\delta {\bf B}$ equal to zero.
Using the notation ${\hat V}^{\bf x}\equiv\boldsymbol{\eta}$, 
(\ref{eq:noeth3}) reduce to:
\beqnar
&&\nabla{\bf\cdot}(\rho {\hat V}^{\bf x})=0,
\quad {\hat V}^{\bf x}{\bf\cdot}\nabla S=0,\nonumber\\
&&\nabla\times\left({\hat V}^{\bf x}\times{\bf B}\right)
=0,
\label{eq:noeth13}
\eeqnar
where we used Gauss's law  $\nabla{\bf\cdot}{\bf B}=0$. 
Setting $\delta {\bf u}=0$ 
in (\ref{eq:noeth12}) gives the equation:
\beqn
\frac{d{\hat V}^{\bf x}}{dt}-{\hat V}^{\bf x}{\bf\cdot}\nabla{\bf u}=0, 
\label{eq:noeth15}
\eeqn
where $d/dt=\partial/\partial t+{\bf u}{\bf\cdot}\nabla$ is the Lagrangian time 
derivative moving with the flow.  The condition (\ref{eq:noeth15}) 
shows that the vector field ${\hat V}^{\bf x}$ is Lie dragged with the 
flow. 

\subsubsection{Noether's 2nd theorem: mass conservation symmetry}

Consider the conservation law associated 
with the mass conservation equation for the case of an ideal, 
isobaric fluid, with equation of state $p=p(\rho)$ 
(see also Cotter and Holm (2012)).   
For Noether's second theorem the variation of $J$, $\delta J$, 
is given by (\ref{eq:noeth6d}), i.e. we require:
\beqn
\delta J=\int d^3x\ \int\ dt 
\left[\boldsymbol{\eta}{\bf\cdot}{\bf E}(\ell)
+\deriv {D}{t}+\nabla{\bf\cdot}{\bf F}\right]=0, \label{eq:noethd1}
\eeqn
where ${\bf E}(\ell)$ is the Euler operator given by (\ref{eq:noeth6b}). For 
the fluid relabeling symmetry for mass conservation, 
the variation $\delta a$ of $a=\rho d^3x$ is set 
equal to zero, i.e.,
\beqn
\delta a=-{\cal L}_\eta(\rho d^3x)=0. \label{eq:n4}
\eeqn
 Using Cartan's magic formula:
\beqn
{\cal L}_\eta (a)=d(\eta\mathrel{\lrcorner} a)+\eta\mathrel{\lrcorner} da,
\label{eq:n5}
\eeqn
$da=0$ (as $a$ is a three-form in 3D-space), and noting  
 $\eta\mathrel{\lrcorner} a=\rho\eta{\bf\cdot}d{\bf S}$, we obtain 
\beqn
{\cal L}_\eta(\rho d^3x)=d[\rho\eta{\bf\cdot}d{\bf S}]=0.
\label{eq:n6}
\eeqn

By the Poincar\'e Lemma, there exists a 1-form
$\boldsymbol{\psi}{\bf\cdot}d{\bf x}$
such that
\beqn
\eta{\mathrel{\lrcorner}}a=\rho\eta{\bf\cdot}d{\bf S}
=d(\boldsymbol{\psi}{\bf\cdot}d{\bf x})
\equiv \nabla\times{\boldsymbol\psi}{\bf\cdot}d{\bf S}.
\label{eq:n7}
\eeqn
Since $\eta{\mathrel{\lrcorner}}a$ is a conserved advected 2-form,
then
\beqn
\eta=\frac{\nabla\times{\boldsymbol\psi}}{\rho} \quad\hbox{is a conserved
(Lie dragged) vector field}.  \label{eq:n8}
\eeqn
A simpler derivation of (\ref{eq:n8}) is to  note that 
$\boldsymbol{\eta}\equiv {\hat V}^{\bf x}$ satisfies the first 
Lie determining equation in
(\ref{eq:noeth13}), i.e.  $\nabla{\bf\cdot}(\rho\boldsymbol{\eta})=0$.

The first term in (\ref{eq:noethd1}) containing the Euler operator
 ${\bf E}(\ell)$ is:
\beqnar
T_1&=&\int d^3x\int dt\ {\boldsymbol{\eta}}{\bf\cdot}{\bf E}(\ell)
=\int d^3x\ \int dt\ \frac{\nabla\times\boldsymbol{\psi}}{\rho}
{\bf\cdot} {\bf E}(\ell)\nonumber\\
&=&\int d^3x\int dt\ \left\{\nabla{\bf\cdot}
\left[\boldsymbol{\psi}\times{\bf E}(\ell)/\rho\right]
+\boldsymbol{\psi}{\bf\cdot}\nabla\times({\bf E}(\ell)/\rho)\right\}
\nonumber\\
&=&\int d^3x \int dt\  {\boldsymbol{\psi}}{\bf\cdot}
\nabla\times({\bf E}(\ell)/\rho), \label{eq:noethd2}
\eeqnar
where the surface term due to $\nabla{\bf\cdot}[{\boldsymbol{\psi}}\times{\bf E}(\ell)/\rho]$ is assumed to vanish on the boundary $\partial V$ of the volume 
$V$ of integration.  

The remaining integrals in $\delta J$ in (\ref{eq:noethd1}):
\beqn
T_2=\int d^3x\int dt\ \left(\deriv{D}{t}+\nabla{\bf\cdot}{\bf F}\right)=C(t)+
\int d^3x\int dt\ \nabla{\bf\cdot}{\bf F}, \label{eq:noethd3}
\eeqn
can be reduced to the form:
\beqn
T_2=\int d^3x\int dt\left\{ \boldsymbol{\psi}{\bf\cdot}
\left[\deriv{\boldsymbol{\omega}}{t}
-\nabla\times({\bf u}\times
\boldsymbol{\omega})\right]+\nabla{\bf\cdot}{\bf W}\right\},
 \label{eq:noethd4}
\eeqn
where 
\beqn
{\bf W}=\nabla\times\left[\left(h+\frac{1}{2}|{\bf u}|^2\right)
\boldsymbol{\psi}
-(\boldsymbol{\psi}{\bf\cdot}{\bf u}){\bf u}\right], 
\label{eq:noethd4a}
\eeqn
and ${\boldsymbol{\omega}}=\nabla\times{\bf u}$ is the 
vorticity of the fluid. Note that $\nabla{\bf\cdot}{\bf W}=0$, 
because ${\bf W}$ may be written in the form of a 'curl': 
${\bf W}=\nabla\times{\bf M}$. Put another way 
\beqn
\int_V \nabla{\bf\cdot}{\bf W}\ d^3x
=\int_{\partial V} \nabla\times{\bf M}{\bf\cdot}d{\bf S}
=\int_{\partial\partial V} {\bf M}{\bf\cdot}d{\bf x}, 
\label{eq:noeth4b}
\eeqn
which is zero since $\partial\partial V$ does not exist 
(i.e. the boundary of a boundary is zero  for 
a simply connected region: e.g. (Misner et al. (1973)). 
Combining (\ref{eq:noethd2}) 
and (\ref{eq:noethd4}) we obtain:
\beqn
\delta J=\int d^3x\int dt\left\{ {\boldsymbol{\psi}}{\bf\cdot}
\left[\deriv{\boldsymbol{\omega}}{t}
-\nabla\times({\bf u}\times{\boldsymbol{\omega}})
+\nabla\times\left(\frac{\bf E(\ell)}{\rho}\right)\right]
+\nabla{\bf\cdot}{\bf W}\right\}. \label{eq:noethd5}
\eeqn
Thus, invoking the du-Bois Reymond lemma of the Calculus of variations 
and noting that $\nabla{\bf\cdot}{\bf W}=0$, 
(\ref{eq:noethd5}) yields the generalized Bianchi identity:
\beqn
\deriv{\boldsymbol{\omega}}{t}
-\nabla\times({\bf u}\times\boldsymbol{\omega})
+\nabla\times\left(\frac{\bf E(\ell)}{\rho}\right)=0. \label{eq:noethd6}
\eeqn
Equation (\ref{eq:noethd6}) is the basic result of Noether's second 
theorem, which shows that there are differential relations between 
the Euler-Lagrange 
variational derivatives $E_i(\ell)$ ($1\leq i\leq 3$) in this case.  
Note that (\ref{eq:noethd6})  does not necessarily imply that the 
Euler Lagrange equations $E_i(\ell)=0$ ($1\leq i\leq 3$) are satisfied. 
In the case where 
$E_i(\ell)=0$ ($1\leq i\leq 3$), (\ref{eq:noethd6}) implies the vorticity 
conservation law:
\beqn
\left(\derv{t}+{\cal L}_{\bf u}\right)(\boldsymbol{\omega}{\bf\cdot}d{\bf S})
=\left(\deriv{\boldsymbol{\omega}}{t}
-\nabla\times({\bf u}\times\boldsymbol{\omega})
+{\bf u}\nabla{\bf\cdot}\boldsymbol{\omega}\right){\bf\cdot}d{\bf S}=0. 
\label{eq:noethd7}
\eeqn
Note that $\nabla{\bf\cdot}\boldsymbol{\omega}=0$ as 
$\boldsymbol{\omega}=\nabla\times{\bf u}$ is the vorticity. 
Equation  (\ref{eq:noethd7}) shows that the vorticity 
2-form $\boldsymbol{\omega}{\bf\cdot}d{\bf S}$ is advected with the flow. 

The generalized Bianchi identity could also be derived using the 
method of Lagrange multipliers for Noether's second theorem developed by 
 Hydon and Mansfield (2011). 
The proof of (\ref{eq:noethd3})-(\ref{eq:noethd4})
is given below.
\begin{proof}
We use the analysis of Cotter and Holm (2012) to calculate 
$C(t)$. Using (\ref{eq:noeth6e}) and (\ref{eq:noethd4}) $C(t)$ is given by: 
\beqnar
C(t)=&&\left\langle \frac{\delta\ell}{\delta {\bf u}},
\boldsymbol{\eta}\right\rangle
=\int_D
\biggl(
\frac{\delta\ell} {\delta{\bf u}} {\bf\cdot}{\boldsymbol{\eta}}
\biggr) d^3 x
\nonumber\\
=&&\int\left(\frac{1}{\rho}
\frac{\delta\ell} {\delta u_j}\right)\rho\eta_j\ d^3x
=\int\frac{1}{\rho} \frac{
\delta\ell} {\delta u_j}(\nabla\times\boldsymbol{\psi})_j dS_j dx_j\nonumber\\
=&&\int\frac{1}{\rho}
\frac{
\delta\ell} {\delta {\bf u}}
{\bf\cdot}d{\bf x}
\wedge
d({\bf\psi\cdot}d{\bf x}).
\label{eq:n9}
\eeqnar
From (\ref{eq:n9})
\beqn
\frac{dC}{dt}=\int\left\{\derv{t}\left(\frac{1}{\rho}
\frac{
\delta\ell} {\delta {\bf u}}
{\bf\cdot}d{\bf x}\right)\wedge
d({\bf\psi\cdot}d{\bf x})
+\frac{1}{\rho}
\frac{\delta\ell} {\delta {\bf u}}
{\bf\cdot}d{\bf x}\wedge
\derv{t}[d({\bf\psi\cdot}d{\bf x})]\right\}. \label{eq:n10}
\eeqn
Write $dC/dt=t_1+t_2$ where $t_1$ is first term and $t_2$
second term in (\ref{eq:n10}). Note that $a$,  $\boldsymbol{\eta}$
and $(\boldsymbol{\eta}{\mathrel{\lrcorner}}a)$, where $a=\rho\ d^3x$
are advected with the flow. Thus,
\beqn
\left(\derv{t}+{\cal L}_{\bf u}\right)(\boldsymbol{\eta}{\mathrel{\lrcorner}}
a)\equiv
\left(\derv{t}+{\cal L}_{\bf u}\right)
[d(\boldsymbol{\psi}{\bf\cdot}d{\bf x})]=0. \label{eq:n11}
\eeqn

At this point it is useful to introduce the notation:
\beqn
\alpha=\frac{1}{\rho}
\frac{\delta\ell} {\delta {\bf u}}{\bf\cdot}d{\bf x},
\quad \beta={\cal L}_{\bf u}(\boldsymbol{\psi}{\bf\cdot}
d{\bf x}), \quad \gamma=\boldsymbol{\psi}{\bf\cdot}d{\bf x}. \label{eq:ld1}
\eeqn

Using the results:
\begin{align}
&d(\alpha\wedge\beta)=d\alpha\wedge\beta-\alpha\wedge d\beta, 
\quad d(d\alpha\wedge\gamma)=0, \nonumber\\
&{\cal L}_{\bf u}(d\alpha\wedge\gamma)={\cal L}_{\bf u}(d\alpha)\wedge\gamma+d\alpha\wedge {\cal L}_{\bf u}(\gamma), \nonumber\\
&{\cal L}_{\bf u}(d\alpha\wedge\gamma)
={\bf u}\contr d(d\alpha\wedge\gamma)
+d[{\bf u}\contr(d\alpha\wedge\gamma)], \nonumber\\
&\alpha_t\wedge d\gamma=d\alpha_t\wedge\gamma-d(\alpha_t\wedge\gamma),
\label{eq:ld3a}
\end{align}
we obtain:
\beqn
\frac{dC}{dt}=\int\left\{\left(\derv{t}
+{\cal L}_{\bf u}\right)(d\alpha)\wedge\gamma
+d[\alpha\wedge\beta-{\bf u}\contr(d\alpha\wedge\gamma)
-\alpha_t\wedge\gamma]\right\}. \label{eq:ld4}
\eeqn
Using (\ref{eq:ld4}) for $dC/dt$ in (\ref{eq:noethd1}) for $\delta J$ gives:
\begin{align}
\delta J=&\int dt\ \biggl\{\left(\derv{t}
+{\cal L}_{\bf u}\right)(d\alpha)\wedge\gamma
+d\biggl[\boldsymbol{\psi}\times\frac{{\bf E}(\ell)}{\rho}{\bf\cdot} d{\bf S}
+{\bf F}{\bf\cdot}d{\bf S}\nonumber\\
&+\alpha\wedge\beta-{\bf u}\contr(d\alpha\wedge\gamma)
-\alpha_t\wedge\gamma\biggr]\biggr\}
+\int d^3x\int dt\ \boldsymbol{\psi}
{\bf\cdot}\nabla\times
\left({\bf E}(\ell)/\rho\right). \label{eq:ld5}
\end{align}

Next we note that the surface term:
\begin{align}
&d\left[{\bf F}{\bf\cdot}d{\bf S}+\alpha\wedge\beta
-{\bf u}\contr(d\alpha\wedge\gamma)
-\alpha_t\wedge\gamma\right]\nonumber\\
&=d({\bf W}{\bf\cdot}d{\bf S})=\nabla{\bf\cdot}{\bf W} d^3x, \label{eq:ld6}
\end{align}
where
\beqn
{\bf W}=\nabla\times\left[\left(h+\frac{1}{2}|{\bf u}|^2\right) 
\boldsymbol{\psi}-(\boldsymbol{\psi}{\bf\cdot}{\bf u}){\bf u}\right]. 
\label{eq:ld7}
\eeqn
Note that $\nabla{\bf\cdot}{\bf W}=0$. 
In (\ref{eq:ld7}) we assumed a barotropic equation of state,
with $p=p(\rho)$, and used the momentum equation:
\beqn
{\bf u}_t-{\bf u}\times\boldsymbol{\omega}
+\nabla\left(\frac{1}{2} |{\bf u}|^2\right)
=T\nabla S-\nabla h,  \label{eq:ld7a} 
\eeqn
to determine $\alpha_t$.  
Also note that
\begin{align}
&\int\left(\derv{t}+{\cal L}_{\bf u}\right)(d\alpha)\wedge\gamma
=\int\left(\derv{t}+{\cal L}_{\bf u}\right)
(\boldsymbol{\omega}{\bf\cdot}d{\bf S})
\wedge(\boldsymbol{\psi}{\bf\cdot}d{\bf x})\nonumber\\
&=\int \boldsymbol{\psi}{\bf\cdot}
\left[\boldsymbol{\omega}_t
-\nabla\times({\bf u}\times\boldsymbol{\omega})\right]\ d^3x. \label{eq:ld8}
\end{align}
Substituting (\ref{eq:ld6})-(\ref{eq:ld8}) into  
 (\ref{eq:ld5}), and assuming the surface term due to 
$\boldsymbol{\psi}\times{\bf E}(\ell)/\rho$ is zero, 
we obtain the result (\ref{eq:noethd5}) for $\delta J$.
 This completes the proof.
\end{proof}

\subsubsection{Cross helicity}

To obtain the cross helicity conservation law (\ref{eq:gch3}) using Noether's 
theorem, it is neccesary to obtain the appropriate solution 
of (\ref{eq:noeth11})-(\ref{eq:noeth15}) 
for the fluid relabeling 
symmetries. The condition that the mass 3-form $\alpha=\rho d^3x$ is 
a fluid relabeling symmetry using Cartan's magic formula, and noting 
$d\alpha=0$ requires that:
\beqn
{\cal L}_{\boldsymbol{\eta}}(\rho d^3x)=d(\boldsymbol{\eta}\mathrel{\lrcorner}
\rho d^3x)=d(\rho\boldsymbol{\eta}{\bf\cdot}d{\bf S})
=\nabla{\bf\cdot}(\rho\boldsymbol{\eta})d^3x=0. \label{eq:ch1}
\eeqn
The entropy variation $\delta S=-\boldsymbol{\eta}{\bf\cdot}\nabla S=0$, 
and the magnetic field variation $\delta {\bf B}=\nabla\times(\boldsymbol{\eta}\times {\bf B})=0$ and the fluid velocity variation $\delta {\bf u}
=\dot{\boldsymbol{\eta}}+[{\bf u},\boldsymbol{\eta}]=0$ are all 
satisfied by the choice:
\beqn
\boldsymbol{\eta}\equiv \hat{V}^{\bf x}
=\zeta({\bf x}_0){\bf b}\quad\hbox{where}
\quad {\bf b}=\frac{\bf B}{\rho}\quad 
\hbox{and}\quad {\bf B}{\bf\cdot}\nabla S=0. \label{eq:ch2}
\eeqn
Note that ${\bf b}={\bf B}/\rho$ is an invariant vector field that is 
Lie dragged with the flow (see (\ref{eq:chn3})). From
(\ref{eq:noeth6f}) the surface term $D$  
in the variational principle  (\ref{eq:noeth6a}) is given by:
\beqn
D=\rho {\bf u}{\bf\cdot}\boldsymbol{\eta}+\Lambda^0
=\rho {\bf u}{\bf\cdot}\zeta({\bf x}_0) {\bf b}+\Lambda^0
\equiv \zeta({\bf x}_0) {\bf u}{\bf\cdot}{\bf B}+\Lambda^0. \label{eq:ch3}
\eeqn
Similarly, the flux ${\bf F}$ surface term in (\ref{eq:noeth6f}) is given 
by:
\beqnar
{\bf F}&=&\zeta({\bf x}_0)\frac{\bf B}{\rho}{\bf\cdot}
\biggl[\rho {\bf u}\otimes {\bf u}
+\rho\left(h-\frac{1}{2}|{\bf u}|^2\right) {\sf I} 
+\frac{B^2}{\mu_0}{\sf I}-\frac{{\bf B}\otimes{\bf B}}{\mu_0}\biggr]
+\boldsymbol{\Lambda}
\nonumber\\
&=&\zeta({\bf x}_0)\left[({\bf u\cdot B}){\bf u}
+\left(h-\frac{1}{2}|{\bf u}|^2\right){\bf B}\right]+\boldsymbol{\Lambda}. 
\label{eq:ch4}
\eeqnar
In (\ref{eq:ch3}) and (\ref{eq:ch4}) we have added the gauge potential terms 
$\Lambda^0$ and $\boldsymbol{\Lambda}$. This allows one to make a 
link to the variational approach of Section 6 that includes the effects 
of gauge transformations in the variational principle and in Noether's theorem.
In Section 6, the generalized cross helicity conservation 
law (\ref{eq:fh5a}) was 
obtained by setting $\zeta({\bf x}_0)=1$, 
$\Lambda^0=r{\bf B}{\bf\cdot}\nabla S$ and 
$\boldsymbol{\Lambda}={\bf u}r {\bf B}{\bf\cdot}\nabla S$ where $dr/dt=-T$ 
(see equations (\ref{eq:chn2})). 
In the variational principle (\ref{eq:noeth6a}) $\delta J$ reduces to:
\beqnar
\delta J&=&\int d^3x\int dt\
 \biggl\{\zeta({\bf x}_0)\frac{\bf B}{\rho}{\bf\cdot}{\bf E}(\ell) 
+\derv{t}\left(\zeta({\bf x}_0){\bf u}{\bf\cdot}{\bf B}\right)\nonumber\\
&\ &+ \nabla{\bf\cdot}\left[
\zeta({\bf x}_0)
\left(({\bf u\cdot B}){\bf u}
+\left(h-\frac{1}{2}|{\bf u}|^2\right){\bf B}\right)\right]
+\nabla_\alpha\Lambda^\alpha\biggr\}
\nonumber\\
&=&\int d^3x\int dt\ \biggl\{
\zeta({\bf x}_0)
\biggl[\frac{{\bf B\cdot E}(\ell)}{\rho}
+\deriv{h_c}{t}\nonumber\\
&\ &+\nabla{\bf\cdot}\left[{\bf u} h_c +
\left(h-\frac{1}{2}|{\bf u}|^2\right){\bf B}\right]\biggr]
+R\biggr\}, 
\label{eq:ch5}
\eeqnar
where  $h_c={\bf u\cdot B}$ is the cross helicity, and 
\beqn
R=h_c\left(\deriv{\zeta}{t}+{\bf u}{\bf\cdot}\nabla\zeta\right) 
+\left(h-\frac{1}{2}|{\bf u}|^2\right)
{\bf B}{\cdot}\nabla\zeta+\nabla_\alpha\Lambda^\alpha. \label{eq:ch6}
\eeqn
In the case  
${\bf B}{\bf\cdot}\nabla S({\bf x}_0)={\bf B}{\bf\cdot}\nabla \zeta({\bf x}_0)
=0$, and $\Lambda^\alpha=0$ ($\alpha=0,1,2,3$), 
the remainder term in (\ref{eq:ch5}) and (\ref{eq:ch6}) 
$R=0$. The net upshot from (\ref{eq:ch5}) is the generalized Bianchi 
identity:
\beqn
\frac{{\bf B\cdot E}(\ell)}{\rho}
+\deriv{h_c}{t}
+\nabla{\bf\cdot}\left[{\bf u} h_c +
\left(h-\frac{1}{2}|{\bf u}|^2\right){\bf B}\right]
=0. \label{eq:ch7}
\eeqn
Thus, if the Euler Lagrange equations ${\bf E}(\ell)=0$ 
are satisfied, then  (\ref{eq:ch7}) reduces
to the cross helicity conservation equation (\ref{eq:gch3}), i.e.
\beqn
\deriv{h_c}{t}
+\nabla{\bf\cdot}\left[{\bf u} h_c +
\left(h-\frac{1}{2}|{\bf u}|^2\right){\bf B}\right]
=0. \label{eq:ch8}
\eeqn
The only constraint on (\ref{eq:ch8}) is that we 
require ${\bf B}{\bf\cdot}\nabla S=0$. 
If ${\bf B}{\bf\cdot}{\bf n}=0$ on the boundary $\partial V_m$ of the volume 
$V_m$ of interest, then the integral form of the (\ref{eq:ch8}) reduces to 
$dH_c/dt=0$ (see Section 3 for further discussion). 

\subsubsection{Helicity in Fluids}

In a barotropic, ideal fluid in which the pressure $p=p(\rho)$ is independent of the entropy $S$, the helicity density:
\beqn
h_f={\bf u}{\bf\cdot}\boldsymbol{\omega}\quad\hbox{where}\quad \boldsymbol{\omega}=\nabla\times{\bf u}, \label{eq:hel1}
\eeqn
satisfies the conservation law:
\beqn
\deriv{h_f}{t}+\nabla{\bf\cdot}
\left[{\bf u}h_f+\left(h-\frac{1}{2}|{\bf u}|^2\right)
\boldsymbol{\omega}\right]=0. \label{eq:hel2}
\eeqn
This conservation law is the analogue of the cross helicity 
conservation law (\ref{eq:ch8}) where 
${\bf B}\to\boldsymbol{\omega}$ and $h_c\to h_f$

The Lie symmetry associated with the helicity (kinetic helicity) conservation 
equation (\ref{eq:hel2}) is:
\beqn
\boldsymbol{\eta}\equiv\hat{V}^{\bf x}
=\frac{\zeta({\bf x}_0)\boldsymbol{\omega}}{\rho}\quad\hbox{where}\quad 
\boldsymbol{\omega}{\bf\cdot}\nabla\zeta({\bf x}_0)=0. \label{eq:hel5}
\eeqn
One can verify that the solution (\ref{eq:hel5}) satisifies  
the fluid relabelling Lie determining equations 
(\ref{eq:noeth12})-(\ref{eq:noeth15}) with ${\bf B}=0$.  
In particular (\ref{eq:noeth15}) reduces to the vorticity equation:
\beqn
\frac{d}{dt}\left(\frac{\boldsymbol{\omega}}{\rho}\right)=\frac{\boldsymbol{\omega}}{\rho}{\bf\cdot}\nabla{\bf u}\quad\hbox{or}\quad \deriv{\boldsymbol{\omega}}{t}-\nabla\times ({\bf u}\times\boldsymbol{\omega})=0, \label{eq:hel6}
\eeqn
which applies for a barotropic equation of state with $p=p(\rho)$. 
The derivation of the helicity conservation law (\ref{eq:hel2}) using Noether's
 theorem is analogous to the derivation of the cross helicity 
conservation law (\ref{eq:ch8}) except that ${\bf B}\to\boldsymbol{\omega}$ 
and $h_c\to h_f$. 
 
\subsubsection{Potential vorticity and Ertel's theorem}

\begin{proposition}
Ertel's theorem states that in ideal compressible fluid mechanics, 
that the potential vorticity $q=\boldsymbol{\omega}{\bf\cdot}\nabla S/\rho$ 
where $\boldsymbol{\omega}=\nabla\times{\bf u}$ is the fluid vorticity, is 
advected with the flow, i.e. $dq/dt=0$ where 
$d/dt=\partial/\partial t+{\bf u}{\bf\cdot}\nabla$ is the Lagrangian 
time derivative following the flow.
\end{proposition}

The Lie determining equations (\ref{eq:noeth13})-(\ref{eq:noeth15}) admit 
the symmetry:
\beqn
\boldsymbol{\eta}\equiv \hat{V}^{\bf x}
=\frac{\nabla\times(\Phi\nabla S)}{\rho}
=\frac{\nabla\times\boldsymbol{\psi}}{\rho}, \quad\hbox{where}
\quad \boldsymbol{\psi}=\Phi\nabla S,  \label{eq:pv1}
\eeqn
and $\Phi=\Phi({\bf x}_0)$ depends only on the Lagrange labels ${\bf x}_0$
, i.e $\Phi$ is a 0-form Lie dragged by the flow:
\beqn
\frac{d\Phi}{dt}=\deriv{\Phi}{t}+{\bf u}{\bf\cdot}\nabla\Phi=0. \label{eq:pv2}
\eeqn
Note that
\beqn
\boldsymbol{\eta}\contr\rho\ d^3x=\rho\boldsymbol{\eta}{\bf\cdot}d{\bf S}
=\nabla\times\boldsymbol{\psi}{\bf\cdot}d{\bf S}
=d(\boldsymbol{\psi}{\bf\cdot}d{\bf x})=d(\Phi dS). 
\label{eq:pv3}
\eeqn
The condition (\ref{eq:noeth15}) implies 
$\hat{V}^{\bf x}\equiv\boldsymbol{\eta}$ is a Lie dragged vector field 
which satisfies (\ref{eq:noeth12}). Similarly, the 1-form 
$\alpha=\boldsymbol{\psi}{\bf\cdot}d{\bf x}$ is Lie dragged with the flow, i.e.
$\boldsymbol{\psi}$ satisfies the the equation:
\beqn
\deriv{\boldsymbol{\psi}}{t}-{\bf u}\times(\nabla\times\boldsymbol{\psi})
+\nabla({\bf u\cdot}\boldsymbol{\psi})=0. \label{eq:pv4}
\eeqn
Using $\boldsymbol{\psi}=\Phi\nabla S$, (\ref{eq:pv4}) reduces to:
\beqn
\Phi\nabla\left(\frac{dS}{dt}\right)+\nabla S\frac{d\Phi}{dt}=0. 
\label{eq:pv5}
\eeqn
Equation (\ref{eq:noeth15}) is equivalent to the curl of (\ref{eq:pv5}). 
Since $dS/dt=0$, (\ref{eq:pv5}) implies  
 $d\Phi/dt=0$. Note that $\boldsymbol{\psi}{\bf\cdot}d{\bf x}
=\Phi\nabla S {\bf\cdot}d{\bf x}$ are Lie dragged 1-forms and hence 
$\Phi$ is 
necessarily an advected invariant 0-form or function.

\begin{proof}{\bf Ertel's theorem}

To derive Ertel's theorem from Noether's theorem, we require 
$\delta J=0$ in (\ref{eq:noethd1}). From (\ref{eq:ld5}):
\beqn
\delta J=\int dt\int_{V}\left[\left(\derv{t}+{\cal L}_{\bf u}\right) (d\alpha)
\wedge \gamma+d({\bf W}{\bf\cdot}d{\bf S})\right]
+\int dt\int_V d^3x\ \boldsymbol{\psi}{\bf\cdot}
\nabla\times({\bf E}(\ell)/\rho), \label{eq:pv6}
\eeqn
where ${\bf W}$ is given by (\ref{eq:ld7}). Note that ${\bf W}$ is a 
solenoidal vector field, i.e. $\nabla{\bf\cdot}{\bf W}=0$. In (\ref{eq:pv6}) 
$\boldsymbol{\psi}=\Phi\nabla S$ and $d\Phi/dt=0$. We introduce the notation:
\beqn
I=\int_{V}\left(\derv{t}+{\cal L}_{\bf u}\right) (d\alpha)
\wedge \gamma, \label{eq:pv7}
\eeqn
for the first integral in (\ref{eq:pv6}), where $\alpha$ , $\beta$ and 
$\gamma$ are the differential 1-forms given in (\ref{eq:ld1}). From 
(\ref{eq:pv7}) and (\ref{eq:ld1}) we obtain:
\begin{align}
I=&\int_V\frac{d}{dt}\left[\nabla\times\left(\frac{1}{\rho} \frac{\delta{\ell}}
{\delta{\bf u}}\right){\bf\cdot}d{\bf S}\right]
\wedge\Phi\nabla{S}{\bf\cdot}d{\bf x}\nonumber\\
=&\int_V\frac{d}{dt}(\boldsymbol{\omega}{\bf\cdot}d{\bf S})\wedge \Phi\nabla S{\bf\cdot}d{\bf x}\nonumber\\
=&\int_V\frac{d}{dt}\left(\boldsymbol{\omega}{\bf\cdot}d{\bf S}
\wedge\Phi\nabla S{\bf\cdot}d{\bf x}\right) . \label{eq:pv8}
\end{align}
In (\ref{eq:pv8}) we use the fact that $\Phi$ is a 0-form and 
$\nabla S{\bf\cdot}d{\bf x}$ is a 1-form, which  are Lie dragged 
with the flow. The integral $I$ in (\ref{eq:pv8}) can be further reduced to:
\beqn
I=\int_V\frac{d}{dt}\left(\frac{\boldsymbol{\omega}{\bf\cdot}\nabla S}{\rho} 
\Phi\rho\ d^3x\right)=\int_V\frac{d}{dt}
\left(\frac{\boldsymbol{\omega}{\bf\cdot}\nabla S} {\rho}\right)
\Phi\rho\ d^3x. \label{eq:pv9}
\eeqn
Note that $d/dt(\Phi\rho d^3x)=0$ as $\rho d^3x$ is an invariant 3-form and 
$\Phi$ is an invariant 0-form. 

Using (\ref{eq:pv9}) in (\ref{eq:pv6}) gives:
\beqn
\delta J=\int dt\ \int_V\ d^3x\left\{\Phi\left[\rho\frac{d}{dt}
\left(\frac{\boldsymbol{\omega}{\bf\cdot}\nabla S}{\rho}\right)
+\nabla S{\bf\cdot}\nabla\times\left(\frac{{\bf E}(\ell)}{\rho}\right)\right]
+\nabla{\bf\cdot}{\bf W}\right\}. \label{eq:pv10}
\eeqn
Because $\nabla{\bf\cdot}{\bf W}=0$, and using the du-Bois Reymond lemma
in (\ref{eq:pv10}), we obtain the generalized Bianchi identity:
\beqn
\rho\frac{d}{dt}
\left(\frac{\boldsymbol{\omega}{\bf\cdot}\nabla S}{\rho}\right)
+\nabla S{\bf\cdot}\nabla\times\left(\frac{{\bf E}(\ell)}{\rho}\right)=0. 
\label{eq:pv11}
\eeqn
If the Euler-Lagrange equations ${\bf E}(\ell)=0$ are satisfied, 
then (\ref{eq:pv11}) implies Ertel's theorem:
\beqn
\frac{d}{dt}\left(\frac{\boldsymbol{\omega}{\bf\cdot}\nabla S}{\rho}\right)=0. 
\label{eq:pv12}
\eeqn
This completes the proof.

\end{proof}



\section{Concluding Remarks}
In this paper we have used variants of Noether's theorems to derive nonlocal 
conservation laws for helicity and cross helicity in ideal fluid dynamics
and in MHD. These two conservation laws were derived in Webb et al. (2013b).
 Other conservation laws for advected invariants 
of MHD and ideal gas dynamics were obtained by using Lie dragging techniques 
(Webb et al. (2013b), and Tur and Janovsky (1993)). 
If the gas is isobaric (i.e. the gas pressure $p=p(\rho)$), 
the helicity and cross helicity conservation laws are local conservation laws
that depend only on the local variables $(\rho,{\bf u},S,{\bf B})$. 
Also if $p=p(\rho,S)$ 
and  ${\bf B}{\bf\cdot}\nabla S=0$, the cross helicity conservation law is  
 a local conservation law.  For the general case of a non-isobaric
gas with $p=p(\rho,S)$, nonlocal conservation laws were obtained that 
depend on the non-local Clebsch potentials. 

The connection between advected invariants
 and the Casimir invariants was investigated in Section 3. 
Padhye and Morrison (1996a,b) using the canonical Lie bracket 
for Lagrangian MHD, used the fluid relabelling symmetry equations 
to derive the determining equations for the Casimirs.

The nonlocal helicity and cross helicity conservation laws 
were derived in the present paper by using Clebsch variables in   
Noether's theorem and by exploiting  fluid relabelling symmetries 
and gauge symmetries of the action. The energy conservation law in MHD 
was also derived by using a fluid relabelling symmetry of the 
action and including a non-zero gauge potential in the action. 

An alternative derivation of the helicity conservation laws was 
carried out in Section 7 where the Euler Poincar\'e 
formulation of Noether's first theorem and Noether's second theorem 
was developed similar to that of Cotter and Holm (2012) 
(see Holm et al. (1998) for a general account of the Euler Poincar\'e 
equations and semi-direct product Lie algebras 
applied to Hamiltonian systems). Noether's second theorem 
plays an important role in cases where the variational 
principle admits an infinite class of symmetries. In this case 
the conservation laws involve so-called Bianchi identities, 
since the Euler Lagrange equations are not necessarily independent
(e.g. Hydon and Mansfield (2011), Padhye and Morrison (1996a,b)).  
This approach uses Eulerian variations of the action. 
The use of Lie symmetries for differential equations and 
Noether's theorems are described in standard texts (e.g. Olver (1993), 
Ibragimov (1985), Bluman and Kumei (1989), Bluman et al. (2010)). 
The helicity and cross helicity conservation  
laws for barotropic and non-barotropic equations of state for the gas, 
were derived using Noether's theorems coupled with fluid relabelling symmetries and gauge transformations. One surprising result, was the derivation 
of the energy conservation equation for MHD by using a fluid relabelling 
 symmetry and a gauge transformation for the action.

\section*{Aknowledgements}
GMW acknowledges stimulating discussions of MHD conservation laws, 
Casimirs and Noether's theorems  with
Darryl Holm. {\bf We also acknowledge discussion of the MHD Poisson bracket with Phil. Morrison.} We acknowledge detailed comments by one of the referees 
which lead to an improved presentation. 
QH was supported in
part by NASA grants NNX07AO73G and NNX08AH46G. B. Dasgupta was
supported in part by an internal UAH grant. GPZ was supported
in part by NASA grants
NN05GG83G and NSF grant
nos. ATM-03-17509 and ATM-04-28880.
JFMcK acknowledges support by the NRF of South Africa.

\section*{References}
\begin{harvard}



\item[]
Arnold, V. I. and Khesin, B. A. 1998, Topological Methods in Hydrodynamics, 
Springer, New York.



\item[]
Berger, M. A. and Field, G. B., 1984,
The toplological properties of magnetic helicity, {\it J. Fluid. Mech.},
 {\bf 147}, 133-148.


\item[]
Berger, M. A. and Ruzmaikin, A. 2000, Rate of helicity production by
solar rotation, {\it J. Geophys. Res.}, {\bf 105}, (A5), p. 10481-10490.


\item[]
Bieber, J. W., Evenson, P. A. and Matthaeus, W.H., 1987, Magnetic helicity of
the Parker field, {\it Astrophys. J.}, 315, 700.

\item[]
Bluman, G. W. and Kumei, S. 1989, Symmetries and differential Equations, 
(New York: Springer). 

\item[]
 Bluman, G. W., Cheviakov, A. F. and Anco, S. 2010, Applications 
of Symmetry Methods to Partial Differential Equations, Springer: New York. 

\item[]
Bott, R. and Tu, L. W.
1982, \textit{Differential Forms in Algebraic Topology}, Springer: New York.







\item[]
Chandre, C., de Guillebon, L., Back, A. Tassi, E. and Morrison, P. J. 2013, On the use of projectors for Hamiltonian systems and their relationship with 
Dirac brackets, {\it J. Phys. A, Math. and theoret.}, {\bf 46}, 125203 (14pp), 
doi:10.1088/1751-8113/46/12/125203.

\item[]
Cotter, C. J., Holm, D. D., and Hydon, P. E., 2007, Multi-symplectic 
formulation of fluid dynamics using the inverse map, 
{\it Proc. Roy. Soc. Lond.} A, {\bf 463}, 2617-2687 (2007).

\item[]
Cotter, C. J. and Holm, D.D. 2012, On Noether's theorem for 
Euler Poincar\'e equation on the diffeomorphism group 
with advected quantities, {\it Foundations 
of Comput. Math.}, doi:10.1007/S10208-012-9126-8


\item[]
Dewar, R. L. 1970, Interaction between hydromagnetic waves and a time 
dependent inhomogeneous medium, {\it Phys. Fluids}, {\bf 13} (11), 2710-2720.






\item[]
Finn, J. H. and Antonsen, T. M. 1985,  Magnetic helicity: what is it and what is it good for?, {\it Comments on Plasma Phys. and Contr.
Fusion}, {\bf 9}(3), 111.

\item[]
Finn, J. M. and Antonsen, T. M. 1988, Magnetic helicity injection for
configurations with field errors, {\it Phys. Fluids}, {\bf 31} (10), 3012-3017.







\item[]
Hameiri, E. 2003, Dynamically accessible perturbations and 
magnetohydromagnetic stability, {\it Physics of Plasmas}, {\bf 10}, No. 7, 
2643-2648, doi:1070-664X/2003/10(7)/2643/6.

\item[]
Hameiri, E. 2004, The complete set of Casimir constants of the motion in 
magnetohydrodynamics, {\it Physics of Plasmas}, {\bf 11}, No. 7, 3423-3431, 
doi:1070-664X/11(7)/3423/9.







\item[]
Holm, D. D., Marsden, J.E., Ratiu, T., and Weinstein, A., 1985, 
Nonlinear stability 
of fluid and plasma equilibria, {\it Phys. Reports,}, 
Review section of Phys. Rev. Lett., 
{\bf123}, Nos. 1 and 2, 1-116,doi: 0370-1573/85.

\item[] 
Holm, D. D. and Kupershmidt, B. A. 1983a, Poisson brackets and Clebsch 
representations for magnetohydrodynamics, multi-fluid plasmas and 
elasticity, {\it Physica D}, {\bf 6D}, 347-363.

\item[]
Holm, D. D. and Kupershmidt, B. A. 1983b, noncanonical Hamiltonian 
formulation of ideal magnetohydrodynamics, {\it Physica D}, {\bf 7D}, 330-333.



\item[]
Holm, D.D., Marsden, J.E. and Ratiu, T.S. 1998, The Euler-Lagrange equations and semi-products with application to continuum theories, {\it Adv. Math.},
{\bf 137}, 1-81.

\item[]
Hollmann, G. H. 1964, {\it Arch. Met., Geofys. Bioklimatol.}, A14, 1.


\item[]
Hydon, P. E. and Mansfield, E. L. 2011, Extensions of Noether's second theorem: from continuous to discrete systems, {\it Proc. Roy. Soc. A}, 
{\bf 467}, pp. 3206-3221, doi:10.1098/rspa.2011.0158.

\item[]
Ibragimov, N. H. 1985,
 {\it Transformation Groups Applied to Mathematical Physics},
 Dordrecht:Reidel.



\item[]
Kamchatnov, 1982, Topological soliton in magnetohydrodynamics, 
{\it Sov. Phys.}, {\bf 55}, (1), 69-73.

\item[]
Kats, A. V. 2003, Variational principle in canonical variables, Weber 
transformation and complete set of local integrals of motion 
for dissipation-free magnetohydrodynamics, {\it JETP Lett.}, {\bf 77}, 
No. 12, 657-661.


\item[]
Kruskal, M. D. and Kulsrud, R. M. 1958, Equilibrium of a magnetically confined 
plasma in a toroid, {\it Phys. fluids}, {\bf 1}, 265.

\item[]
Kuznetsov, E. A. 2006, Vortex line representation for hydrodynamic type
equations, {\it J. Nonl. Math. Phys.}, {\bf 13}, No. 1, 64-80.

\item[]
Kuznetsov, E. A. and Ruban, V.P. 1998, Hamiltonian dynamics of vortex lines
in hydrodynamic type systems, {\it JETP Lett.}, {\bf 67}, No. 12, 1076-1081.

\item[]
Kuznetsov, E.A. and Ruban, V.P. 2000, Hamiltonian dynamics of vortex and
magnetic lines in hydrodynamic type systems, {\it Phys. Rev. E}, {\bf 61},
No. 1, 831-841.

\item[]
Longcope, D. W. and Malanushenko, A. 2008, Defining and calculating
self-helicity in coronal magnetic fields, {\it Astrophys. J.}, {\bf 674},
1130-1143.

\item[]
Low, B. C. 2006, Magnetic helicity in a two-flux partitioning of an ideal hydromagnetic fluid, {\it Astrophys. J.}, {\bf 646}, 1288-1302.

\item[]
Marsden J. E. and Ratiu T.S. 1994, Introduction to Mechanics and Symmetry,
New York,: Springer Verlag.


\item[]
Misner, C.W., Thorne, K.S. and Wheeler, J.A. 1973, {\it Gravitation},
San Francisco: W.H. Freeman.

\item[]
Moffatt, H. K. 1969, The degree of knottedness of tangled vortex lines,
{\it J. Fluid. Mech.}, {\bf 35}, 117.

\item[]
Moffatt,  H. K.  1978, \textit{Magnetic field Generation in Electrically
Conducting Fluids}, Cambridge Univ. Press, Cambridge U.K.

\item[]
Moffatt, H.K. and Ricca, R.L. 1992, Helicity and the Calugareanu invariant,
{\it Proc. Roy. Soc. London, Ser. A}, {\bf 439}, 411.

\item[]
Moiseev, S. S., Sagdeev, R. Z., Tur, A. V. and Yanovsky, V. V. 1982,
On the freezing-in integrals and Lagrange invariants in hydrodynamic models,
{\it Sov. Phys. JETP}, {\bf 56} (1), 117-123.

\item[]
Morrison, P. J. 1998, Hamiltonian description of the ideal fluid, 
{\it Rev. Mod. Phs.}, {\bf 70}, No. 2, 467-521.

\item[]
Morrison, P. J. 1982, Poisson brackets for fluids and plasmas, in Mathematical 
Methods in Hydrodynamics and Integrability of dynamical Systems,
 {\it AIP Proc. Conf.}, {\bf 88}, ed M. Tabor and Y. M. Treve, pp 13-46.

\item[]
Morrison, P.J. and Greene, J.M. 1980, Noncanonical Hamiltonian density
formulation of hydrodynamics and ideal magnetohydrodynamics,
 {\it Phys. Rev. Lett.}, {\bf 45}, 790-794.

\item[]
Morrison, P.J. and Greene, J.M. 1982, Noncanonical Hamiltonian density
formulation of hydrodynamics and ideal magnetohydrodynamics, (Errata),
 {\it Phys. Rev. Lett.}, {\bf 48}, 569.

\item[]
Newcomb, W. A. 1962, Lagrangian and Hamiltonian methods 
in magnetohydrodynamics, {\it Nucl. Fusion Suppl.}, Part 2, 451-463. 


\item[]
Olver, P. J. 1993,
\textit{Applications of Lie groups to Differential Equations}, 2nd Edition
(New York: Springer).

\item[]
Padhye, N.S. 1998, Topics in Lagrangian and Hamiltonian Fluid Dynamics: Relabeling Symmetry and Ion Acoustic Wave Stability, {\it Ph. D. Dissertation},
University of Texas at Austin.

\item[]
Padhye, N. S. and Morrison, P. J. 1996a, Fluid relabeling symmetry,
 {\it Phys. Lett. A}, {\bf 219}, 287-292.

\item[]
Padhye, N. S. and Morrison, P. J. 1996b, Relabeling symmetries in hydrodynamics 
and magnetohydrodynamics, {\it Plasma Phys. Reports}, {\bf 22}, 869-877.

\item[]
Parker, E. N. 1979, Cosmic Magnetic Fields, Oxford Univ. Press, New York.

\item[]
Powell, K. G., Roe, P.L., Linde, T.J., Gombosi, T. I., and De Zeeuw, D.
1999, A solution adaptive upwind scheme for ideal magnetohydrodynamics,
{J. Comput. Phys.} {\bf 154}, 284-309.







\item[]
Semenov, V. S., Korvinski, D. B., and Biernat, H.K. 2002, Euler potentials for 
the MHD Kamchatnov-Hopf soliton solution, {\it Nonlinear Processes in 
Geophysics}, {\bf 9}, 347-354. 


\item[]
Tur, A. V. and Yanovsky, V. V. 1993, Invariants in disspationless 
hydrodynamic media, {\it J. Fluid Mech.}, {\bf 248}, 
Cambridge Univ. Press, p67-106.





\item[]
Webb, G. M., Zank, G.P., Kaghashvili, E. Kh. and Ratkiewicz, R.E. 2005a, 
Magnetohydrodynamic waves in non-uniform flows I: a variational approach, 
{\it J. Plasma Phys.}, {\bf 71}(6), 785-809, doi:10.1017/S00223778050003739.

\item[]
Webb, G. M., Zank, G.P., Kaghashvili, E. Kh and Ratkiewicz, R.E., 2005b,
Magnetohydrodynamic waves in non-uniform flows II: stress energy tensors, 
conservation laws and Lie symmetries,
 {\it J. Plasma Phys.}, {\bf 71}, 811-857, doi: 10.1017/S00223778050003740.

\item[]
Webb, G. M. and Zank, G.P. 2007, Fluid relabelling symmetries, Lie point 
symmetries and the Lagrangian map in magnetohydrodynamics and gas dynamics, 
{\it J. Phys. A., Math. and Theor.}, {\bf 40}, 545-579, 
doi:10.1088/1751-8113/40/3/013.



\item[]
Webb, G. M., Hu, Q., Dasgupta, B., and Zank, G.P. 2010a,
Homotopy
formulas for the magnetic vector potential and magnetic helicity:
The Parker spiral interplanetary magnetic field and magnetic flux ropes,
{\it J. Geophys. Res.}, (Space Physics), {\bf 115}, A10112,
doi:10.1029/2010JA015513; Corrections: {\it J. Geophys. Res.}, {\bf 116}, 
A11102, doi:10.1029/2011JA017286, 22nd November 2011.

\item[]
Webb, G. M., Hu, Q., Dasgupta, B., Roberts, D.A., and Zank, G.P. 2010b,
Alfven simple waves: Euler potentials and magnetic helicity,
 {\it Astrophys. J.}, {\bf 725}, 2128-2151,
doi:10.1088/0004-637X/725/2/2128.

\item[]
Webb, G. M., Hu, Q., McKenzie, J.F., Dasgupta, B. and Zank, G.P. 2013a, 
Advected invariants in MHD and gas dynamics, {\it 12th Internat. Annual 
Astrophysics Conf.}, Myrtle Beach SC, April 15-19, 2013, (submitted). 

\item[]
Webb, G. M., Dasgupta, B., McKenzie, J.F., Hu, Q., 
and Zank, G.P. 2013b, Local and nonlocal advected invariants and helicities in MHD and gas dynamics I: Lie dragging approach, {\it J. Phys. A, Math. and theoret.}, submitted June 19, 2013 (paper I). 
 



\item[]
Woltjer, L. 1958, A theorem on force-free magnetic fields,
{\it Proc. Nat. Acad. Sci.}, {\bf 44}, 489.

\item[]
Yahalom, A. 2013, Ahronov-Bohm effect in magnetohydrodynamics,
 {\it Phys. Lett. A}, {\bf 377}, 1898-1904.

\item[]
Yahalom, A., and Lynden-Bell D. 2008, Simplified variational principle for 
barotropic magnetohydrodynamics, {\it J. Fluid Mech.}, {\bf 607}, 235-265. 


\item[] 
Zakharov, V. E. and Kuznetsov, E.A. 1997, Hamiltonian formalism for 
nonlinear waves, {\it Physics-Uspekhi}, {\bf 40}, (11), 1087-1116.
\end{harvard}
\end{document}